\documentclass{article}
\usepackage{slashbox}

\usepackage{amsfonts,amsmath,amssymb,amsthm}
\usepackage{verbatim,float,url,enumerate}
\usepackage{graphicx,subfigure,epsfig,psfrag}
\usepackage{bm,dsfont,color,appendix}
\usepackage{adjustbox}
\usepackage{mathtools}
\usepackage{bbm}
\usepackage{natbib}
\usepackage{latexsym,graphicx}
\usepackage{setspace}
\usepackage{tikz}
\usetikzlibrary{arrows,positioning}
\usepackage[margin=1 in]{geometry}
\usepackage{rotating}
\usepackage[linesnumbered,ruled]{algorithm2e}
\floatstyle{ruled}
\newfloat{algorithm}{tbhp}{loa}
\floatname{algorithm}{Algorithm}
\usepackage{hyperref}
\usepackage{makecell}
\usepackage{tabularx}
\newtheorem{theorem}{Theorem}
\newtheorem{lemma}{Lemma}
\newtheorem{remark}{Remark}
\newtheorem{corollary}{Corollary} 
\newtheorem{proposition}{Proposition}

\newtheorem{assumptions}{Assumption}

\def\real{{\mathbb R}}

\def\Err{\rm{Err}}

\def\rP{\rm P}
\def\bP{\bold P}
\def\bSigma{\bold \Sigma}

\date{}
\author{Leying Guan\thanks{Dept. of  Statistics,
    Stanford Univ, leying.guan@gmail.com}, Rob Tibshirani \thanks{Depts. of Biomedical Data Sciences, and Statistics,
    Stanford Univ, tibs@stanford.edu}}
\title{Prediction and outlier detection in  classification problems}
\begin{document}
\maketitle
\begin{abstract}
We consider the multi-class classification problem when the training data and the out-of-sample  test data may have different distributions and propose a method  called BCOPS (balanced and conformal optimized prediction sets). BCOPS constructs a prediction set $C(x)$ as a subset of class labels, possibly empty. It tries to optimize the out-of-sample performance, aiming to include the correct class as often as possible, but also detecting outliers  $x$, for which the method returns no prediction (corresponding to $C(x)$ equal to the empty set). The proposed method combines supervised-learning algorithms with the method of conformal prediction to minimize a misclassification loss averaged over the  out-of-sample distribution. The constructed prediction sets  have a finite-sample coverage guarantee without distributional assumptions.  

We also propose a method to estimate the outlier detection rate of a given method. We prove asymptotic consistency and optimality of our proposals under suitable assumptions and illustrate our methods on real data examples. 
\end{abstract}

\section{Introduction}
\label{sec:introduction}
We consider the multi-class classification problem where the training data and the test data may be mismatched. That is, the training data and the test data may have different distributions. We assume the access to the labeled training data and unlabeled test data. Let $\{(x_i, y_i), i = 1,\ldots, n\}$ be the training data set, with continuous features $x_i\in \real^p$ and response $y_i\in \{1,\ldots, K\}$ for $K$ classes.  

In classification problems, one usually aims to produce a good classifier using the training data  that predicts the class $k$ at each $x$. Instead, here we construct  a {\em prediction set }$C(x)$ at each $x$ by solving an optimization problem minimizing the out-of-sample loss directly. The prediction set $C(x)$ might contain multiple labels or be empty. When $K = 2$, for example, $C(x)\in \{\{1\}, \{2\}, \{1,2\}, \emptyset\}$. If $C(x)$ contains multiple labels, it would indicate that $x$ could be any of the listed classes. If $C(x) = \emptyset$, it would indicate that $x$ is likely to be far from the training data and we could not assign it to any class and consider it as an outlier.

There are many powerful supervised learning algorithms that try to estimate $P(y = k|x)$, the conditional probability of $y$ given $x$,  for $k = 1,\ldots, K$. When the test data and the training data have the same distribution, we can often have reasonably good performance and a relatively faithful evaluation of the out-of-sample performance using sample slitting of the training data. However,  the posterior probability $P(y=k|x)$ may not reveal the fact that the training and test data are mismatched. In particular, when erroneously applied to mismatched data, the standard approaches may yield predictions for $x$ far from the training samples, where it is usually better to not make a prediction at all. Figure \ref{fig:illustration1} shows a two dimensional illustrative example. In this example, we have a training data set with two classes and train a logistic regression model with it. The average misclassification loss based on sample splitting of the  training data is extremely low.  The test data comes from a very different distribution. We plot the training data in the upper left plot: the black points represent class 1 and blue points represent class 2, and plot the test data in the upper right plot using red points. The black and blue dashed curves in these two plots are the boundaries for $P(y = 1|x) = 0.05$ and $P(y = 1|x) = 0.95$ from the logistic regression model. Based on the predictions from the logistic model, we are quite confident that the majority of the red points are from class 1. However, in this case, since the test samples are relatively far from the training data, most likely, we consider them to be outliers and don't want to make predictions.

\begin{figure}
\caption{\em Illustrative example I. We show the training data in the upper left plot,  the black points represent class 1 and blue points represent class 2, the black and blue dashed curves are the boundaries for $P(y = 1|x) = 0.05$ and $P(y = 1|x) = 0.95$ based on the logistic model. In the upper right plot, we use the red points to represent the test samples. The black and blue dashed lines in the upper half  of the figure are the decision boundaries for posterior probability of  class 2 being 0.05 and 0.95 based on the logistic regression model. The lower half of Figure \ref{fig:illustration1}, the interior of the dashed curves represent the  density-level sets  achieving 95\% coverage for class 1 and class 2 respectively.}
\label{fig:illustration1}
\begin{center}
\includegraphics[width = .8\textwidth, height = .8\textwidth]{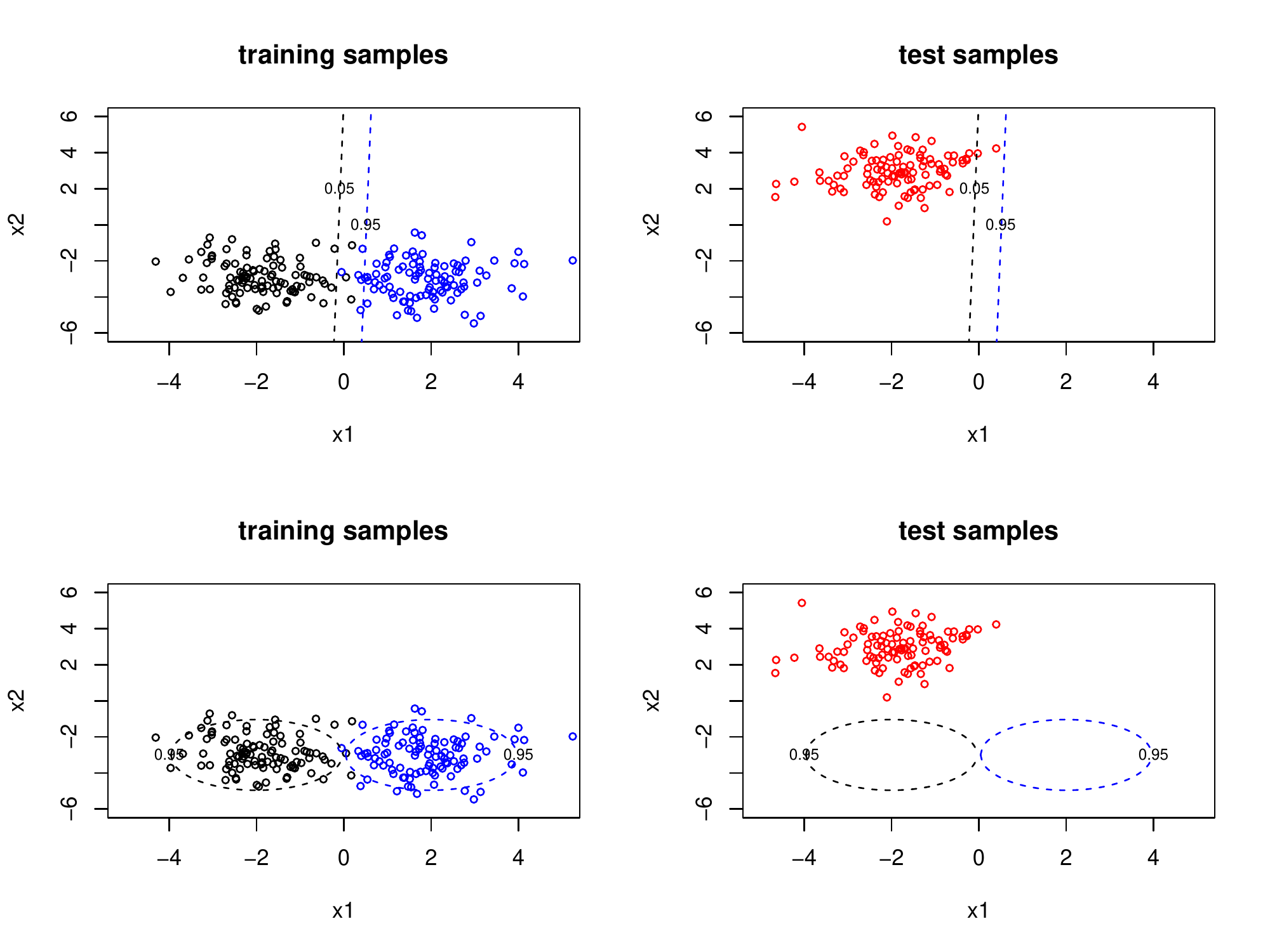} 
\end{center}
\end{figure}

As an alternative, the density-level set \citep{ lei2013distribution, hartigan1975clustering, cadre2006kernel, rigollet2009optimal} considers $f_y(x)$, the density of $x$ given $y$. For each new sample $x$, it constructs a prediction set $C(x)=\{k: x\in A_k\}$ where $A_k = \{x| f_k(x) \geq f_{k,\alpha}\}$ and $f_{k,\alpha}$ is the lower $\alpha$ percentile of $f_k(x)$ under the distribution of class $k$. In Figure \ref{fig:illustration1}, the lower half shows the result of the density-level set with $\alpha = 0.05$. Again, the lower left plot contains the training samples and the lower right plot contains the test samples. The black and blue dashed ellipses are the boundaries for the decision regions $A_1$ and $A_2$ from the the oracle density-level sets (with given densities). We call the prediction with $C(x) = \emptyset$ as the abstention. In this example, we can see that the oracle density-level sets have successfully abstained from predictions while assigning correct labels for most training samples.  The density-level set is also suggested as a way to making prediction with abstention in \cite{hechtlinger2018cautious}. 

However, the density-level set has its own drawbacks. It does not try to utilize information comparing different classes, which can potentially lead to a  large deterioration in performance. Figure \ref{fig:illustration2} shows another example where the oracle density-level set has less than ideal performance. In this example, we have two classes with $x\in \real^{10}$, and the two classes are well separated in the first  dimension and follows the standard normal distribution in other dimensions. In Figure \ref{fig:illustration2}, we show only the first two dimensions. In the the left plot of Figure \ref{fig:illustration2}, we have colored the samples based on their actual class. The black points represent class 1 and the blue points represent class 2. In the right plot of Figure \ref{fig:illustration2}, we have colored the data based on their oracle density-level set results: $x$ is colored green if $C(x) = \{1,2\}$, black if $C(x) = \{1\}$, blue if $C(x) =\{2\}$ and red if $C(x) = \emptyset$. Even though class 1 and 2 can be well separated in the first dimension, we still have $C(x)=\{1,2\}$ for a large portion of data, especially for samples from class 2.

\begin{figure}
\caption{\em Illustrative example II.  We have two classes with $x\in \real^{10}$, and the two classes are well separated in the first  dimension and follows standard normal distribution in other dimensions. The left plot of Figure \ref{fig:illustration2} shows the data colored with their actual class. The black points represent class 1 and the blue points represent class 2. The right plot of Figure \ref{fig:illustration2} shows the data with color corresponding to their density-level set: $x$ is colored green if $C(x) = \{1,2\}$, black if $C(x) = \{1\}$, blue if $C(x) =\{2\}$ and red if $C(x) = \emptyset$.}
\label{fig:illustration2}
\begin{center}
\includegraphics[width = .8\textwidth, height = .4\textwidth]{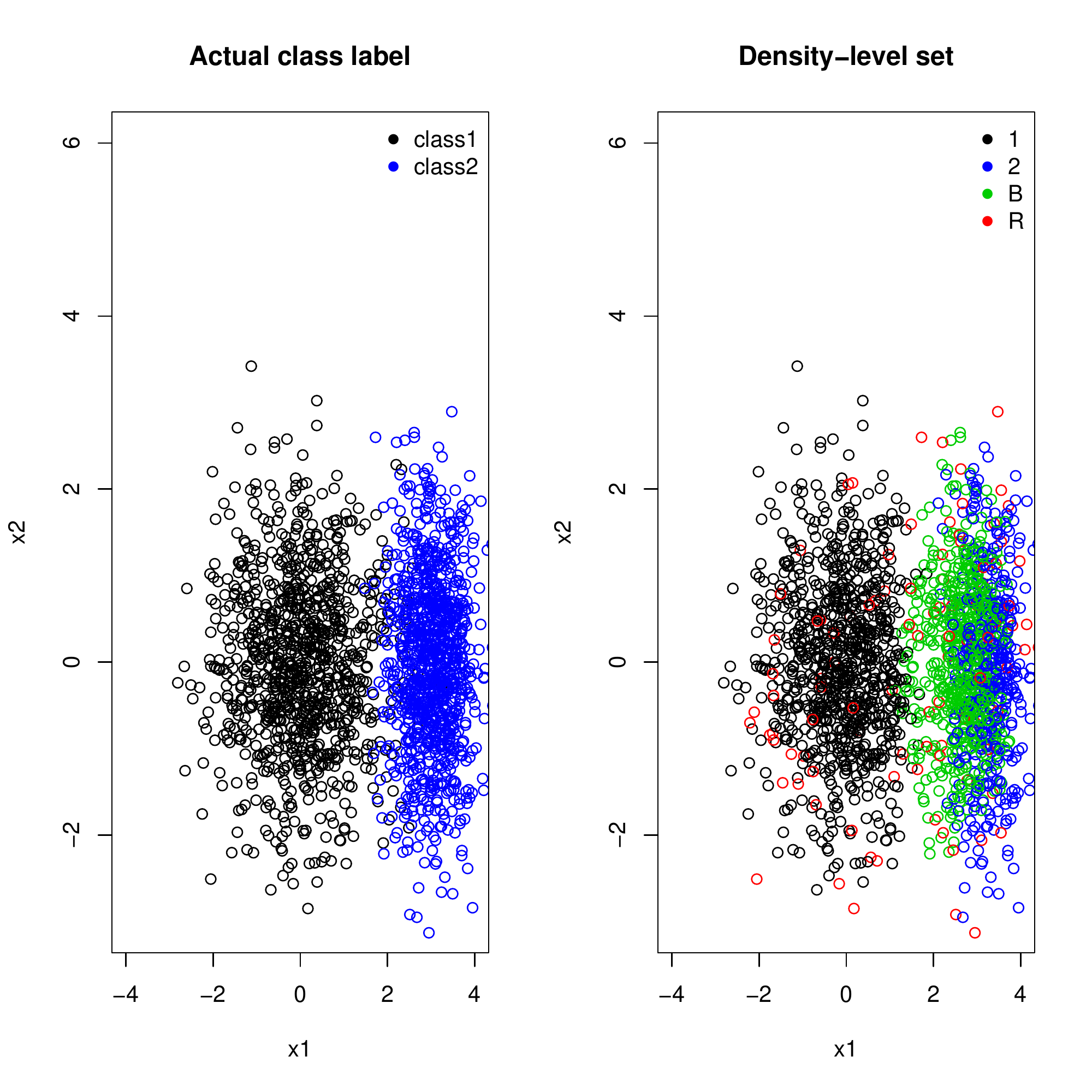} 
\end{center}
\end{figure}

In this paper, following the previous approach of density-level sets, we propose a method  called BCOPS (balanced and conformal optimized prediction set) to construct prediction set $C(x)\subseteq\{1,2,\ldots,K\}$ for each $x$, which tries to make good predictions for samples that are similar to the training data and refrain from making predictions otherwise. BCOPS is usually more powerful than the density-level set because it combines information from different classes when constructing $C(x)$. We also describe a new regression-based method for the evaluation of outlier detection ability under some assumptions on how the test data may differ from the training data. The paper is organized as follows.  In section \ref{sec:methods}, we first describe our model and related works, then we will introduce BCOPS. In section \ref{sec:methods2}, we will describe methods  to evaluate the performance regarding outlier detection. Some asymptotic behaviors of our proposals are given in section \ref{sec:theory}. Finally, we provide real data examples in section \ref{sec:realData}.

\section{BCOPS: Models and methods}
\label{sec:methods}
\subsection{A mixture model}
\label{sec:mehotds_model}
It is often assumed that the distribution of the training data and out-of-sample data are the same.  Let $\pi_k\in (0,1)$ be the proportion samples from class $k\in \{1,\ldots,K\}$ in the training data, with $\sum^K_{k=1} \pi_k = 1$. Let $f_k(x)$ be the density of $x$ from class $k$, and $f(x)$/$f_{test}(x)$ be the marginal in/out-of-sample densities. Under this assumption, we know that
\begin{align}
\label{eq:goal0}
f(x) = \sum^K_{k=1} \pi_k f_{k}(x),\;\;f_{test}(x) = \sum^K_{k=1} \pi_k f_{k}(x).
\end{align}
In particular, $f_{test}(x)$ can be written as a mixture of $f_{k}(x)$, and the mixture proportions $\pi_k$ remain unchanged.  

In this paper, we allow for distributional changes and assume that the out-of-sample data may have different mixture proportions $\tilde \pi_k\in [0,1)$ for $k\in \{1,\ldots, K\}$, and  as well as a new class $\bold{R}$ (outlier class) which is unobserved in the training data. We let

\begin{align}
\label{eq:goal1}
f_{test}(x) = \sum^K_{k=1} \tilde{\pi}_k f_{k}(x)+\epsilon \cdot e(x)
\end{align}
where $e(x)$ is the density of $x$ from the outlier class , $\epsilon\in [0,1)$ is its proportion and $\sum^K_{k=1} \tilde{\pi}_k+\epsilon = 1$.

Under this new model assumption, we want to find a prediction set $C(x)$ that aims to minimize the length of $C(x)$ averaged over a properly chosen importance measure $\mu(x)$, and with guaranteed coverage $(1-\alpha)$ for each class. Let $|C(x)|$ be the size of $C(x)$. We consider the optimization problem $\mathcal{P}$ below:

\begin{align}
\label{eq:conformal1}
&\min \int |C(x)| \mu(x)d x\nonumber \\
s.t.\;\;& P_{k}(k\in C(x)) \geq 1-\alpha, \;\forall k = 1,\ldots, K
\end{align} 
where  $P_{k}(\mathcal{A})$ is the probability of the event $\mathcal{A}$ under the distribution of class $k$ and $\mu(x)$ is a weighting function which we will choose later to tradeoff classification accuracy and outlier detection.  The constraint (3) says that we want the have $k\in C(x)$ for at least $(1-\alpha)$ of samples that are actually from an observed class $k$ (coverage requirement).  If $C(x) =\emptyset$, $x$ is considered to be an outlier at the given level $\alpha$ and we will refrain from making a prediction. 

It is easy to check that problem $\mathcal{P}$ can be decomposed into $K$ independent problems for different classes, referred to as problem $\mathcal{P}_k$:
\begin{align}
\label{eq:conformal2}
&\min \int \mathbbm{1}_{x\in A_k} \mu(x)d x\\
s.t.\;\;& P_{k}(x\in A_k) \geq 1-\alpha
\end{align}
Let $\mathcal{A}_k$ be the solution to problem $\mathcal{P}_k$, then the solution to problem $\mathcal{P}$ is $C(x) = \{k: x\in A_k\}$. The set $\mathcal{A}_k$ has an explicit form using the density functions.

For an event $\mathcal{A}$, let $P_F(\mathcal{A})$ be the probability of $\mathcal{A}$ under distribution $F$ and let $Q(\alpha; g, F)$ be the lower $\alpha$ percentile of a real-valued function $g(x)$ under distribution $F$:
\[
Q(\alpha; g, F) = \sup \{t: P_F(g(x) \leq t) \leq \alpha\}
\]
We use $Q(\alpha; g(x_1), \ldots, g(x_n))$, or  $Q(\alpha; g(x_{1:n}))$ to denote the lower $\alpha$ percentile of $g(x)$ from the empirical distribution using  samples $x_1,\ldots, x_n$. Let $F_k$ be the distribution of $x$ from class $k$. It is easy to check that 
\begin{equation}
\label{eq:BCOPSsolution}
A_k = \{x: v_k(x) \geq Q(\alpha; v_k, F_k)\}, \;\; v_{k}(x) = \frac{f_k(x)}{\mu(x)}
\end{equation}
is the solution to problem $\mathcal{P}_k$. We call this $A_k$ the oracle set for class $k$, the oracle prediction set $C(x)$ for problem $\mathcal{P}$ is constructed using the oracle sets $A_k$. Since $A_k(x)$ depends only on the ordering of $v_k(x)$, we can also use any order-preserving transformation of $v_k(x)$ when constructing $A_k$. (An order-preserving transformation $o:\real\rightarrow \real$ satisfies that $v_1 < v_2 \Leftrightarrow o(v_1) < o(v_2)$ for $\forall v_1, v_2\in \real$.)
\subsection{Strategic choice for the weighting function $\mu(x)$}
\label{sec:methods_weight}
How should we choose $\mu(x)$? For any given $\mu(x)$, while the coverage requirement is satisfied by definition,  the choice of $\mu(x)$ influences how well we separate different observed classes from each other, and inliers from outliers. In practice, except for the coverage, people also want to minimize the  misclassification loss averaged over the out-of-sample data, e.g, 
\begin{equation}
\label{eq:Err}
{\rm Err} = E_{x, y\sim f_{test}(x)}\sum_{k\neq y}\mathbbm{1}_{ k \in C(x)}
\end{equation}
It is easily shown that the solution minimizing the above misclassification loss under the coverage requirement is the same as minimizing the objective $\int |C(x)|\mu(x) dx$ in problem $\mathcal{P}$ with $\mu(x) = f_{test}(x)$. This makes $f_{test}(x)$ a natural choice for $\mu(x)$.

BCOPS constructs $\widehat C(x)$ to approximate the oracle solution $C(x)$ to $\mathcal{P}$ with $\mu(x) = f_{test}(x)$.  Some previous work is closely related or equivalent to other choices of $\mu(x)$. For example, the density-level set described in the introduction can also be written equivalently to the solution of problem $\mathcal{P}$ with $\mu(x)\propto 1$ \citep{lei2013distribution}. 

Our prediction set $\hat C(x)$ is constructed by combining a properly chosen learning algorithms with the conformal prediction idea  to meet the coverage requirement without distributional assumptions \citep{vovk2005algorithmic}. For the remainder of this section, we first have a brief discussion of some related methods in section \ref{sec:methods_related}, and  review the idea of conformal prediction in section \ref{sec:methods_conformal}. We give details of  BCOPS in section \ref{sec:methods_offline},  and we show a simulated example in \ref{sec:methods_example}.  
\subsection{Related work}
\label{sec:methods_related}
The new model assumption described in equation (\ref{eq:goal1})
\begin{align}
f_{test}(x) = \sum^K_{k=1} \tilde{\pi}_k f_{k}(x)+\epsilon\cdot e(x)\nonumber
\end{align}
allows for changes in the mixture proportion, and treats as outliers the part of distributional change that can not be explained. This assumption is different from the assumption that fixes $f(y|x)$ and allows changes in $f(x)$ without constraint. We use this model because $P(y=k)$ is much easier to estimate than $f(x)$, and it explicitly describes what kind of data we would like to reject.

Without the extra term $\epsilon\cdot e(x)$, the change in mixture proportions $\pi_k$ is also called label shift/target shift \citep{zhang2013domain, lipton2018detecting}. \cite{zhang2013domain} also allows for a location-scale transformation in $x|y$. When only label shift happens, a better prediction model can be constructed through sample reweighting using the labeled training data and unlabeled out-of-sample data. 

The extra term $\epsilon\cdot e(x)$ corresponds to the proportion and distribution of outliers. We do not want to make a prediction if a sample comes from the outlier distribution.  There is an enormous literature on the outlier detection problem, and readers who are interested can find a thorough review of traditional outlier detection approaches in \cite{hodge2004survey, chandola2009anomaly}. Here, we go back to to the density-level set. Both BCOPS and the density-level set are based on the idea of prediction sets/tolerance regions/minimum volume sets\citep{wilks1941determination,wald1943extension, chatterjee1980asymptotically, li2008multivariate}, where for each observation $x$, we assign it a prediction set $C(x)$ instead of a single label so as to minimize certain objective, usually the length or volume of $C(x)$, while having some coverage requirements.  As we pointed out before, the density-level set is the optimal solution when $\mu(x)\propto 1$\citep{lei2013distribution, hechtlinger2018cautious}:
\begin{align}
&\min \int |C(x)| d x\nonumber \\
s.t.\;\;& P_{k}(k\in C(x)) \geq 1-\alpha, \;\forall k = 1,\ldots, K\nonumber
\end{align} 
While the density-level prediction set achieves optimality with $\mu(x)$ being the Lebesgue measure, it is not obvious that $\mu(x)\propto 1$ is a good choice. In section \ref{sec:introduction}, we observe that the density-level set has lost the contrasting information between different classes, and choosing $\mu(x)\propto 1$ is a reason for why it happens.
The work of \cite{herbei2006classification, bartlett2008classification} is closely related to the case where $\mu(x)=f(x)$, the in-sample density. When $\mu(x) = f(x)$, we encounter the same problem as in the usual classification methods that learn $P(y|x)$ and could assign confident predictions to test samples which are far from the training data. As a contrast, BCOPS choses $\mu(x)$ to utilize as much information as possible to minimize   ${\rm Err} = E_{x,y\sim f_{test}(x)}\sum_{k\neq y}\mathbbm{1}_{ k \in C(x)}$, which usually leads to not only good predictions for inliers but also abstentions for the outliers.

In an independent recent work,  \cite{barber2019conformal} also used information from the unlabeled out-of-sample data, under a different model and goal.
\subsection{Conformal prediction}
\label{sec:methods_conformal}
BCOPS constructs $\widehat C(x)$ using  the method of conformal prediction. We give a brief recap of the conformal prediction here for completeness. 

Let $X_1,\ldots, X_n\overset{i.i.d}{\sim} P$ and $X_{n+1}$ be a new observation. Conformal prediction considers the question whether $X_{n+1}$ also comes from $\mathcal{P}$ and aims to find a decision rule such that if $X_{n+1}$ is also independently generated from $\mathcal{P}$, than we will accept $X_{n+1}$ (to be from $\mathcal{P}$) with probability at least $1-\alpha$. The key step of the conformal prediction is to construct a real-valued conformal score function $\sigma(\{X_1,\ldots,X_{n+1}\}, x)$ of $x$ that may depend on the observations $\{X_1, \ldots, X_n, X_{n+1}\}$ but is permutation invariant to its first argument.

 Let $\sigma_i = \sigma(\{X_1,\ldots, X_n, X_{n+1}\}, X_i)$.  Then if $X_{n+1}$ is also independently generated from $P$, by symmetry, we have that
\[
s_i = \frac{1}{n+1}\sum^{n+1}_{j=1}\mathbbm{1}_{\sigma_i \geq \sigma_j}
\]
is uniformly distributed on $\{\frac{1}{n+1},\frac{2}{n+1},\ldots, \frac{n}{n+1}, 1\}$ (if there is a tie, breaking it randomly). For any feature value $x$, we decide if  the set $A$ contains  $x$  by letting $X_{n+1} = x$ and consider the corresponding $s_{n+1}$:
\[
A = \{x| s_{n+1} \geq \frac{\lfloor(n+1)\alpha\rfloor}{n+1}\},
\]
Then we have $P(X_{n+1}\in A) \geq 1-\alpha$  \citep{vovk2005algorithmic}.  

The most familiar valid conformal score may be the sample splitting conformal score $\sigma(\{X_1,\ldots, X_n\}, x) = \sigma(x)$ where the conformal score function is independent of the new observation $X_{n+1}$ and observations $\{X_1,\ldots,X_n\}$ that are used to construct $s_{n+1}$ given the conformal score function (but can depend on other training data). From now on, we will call a procedure based on this independence as the sample-splitting conformal construction.

Another simple example where the conformal score function actually relies on the permutation invariance is given below:
\[
\sigma(\{X_1,\ldots, X_n, X_{n+1}\}, x) = -(x-\frac{\sum^{n+1}_{i=1}X_i}{n+1})^2
\]
Since $\sigma(\{X_1,\ldots, X_n, X_{n+1}\}, x) $ is permutation invariant on its first argument $\{X_1,\ldots, X_n, X_{n+1}\}$,  we will have the desired coverage with this score function. We call a procedure of the above type, that relies on the permutation invariance but not the independence between observations and the conformal score function, as the data-augmentation conformal construction.

In BCOPS, we estimate $v_k(x)$ that is used in eq.(\ref{eq:BCOPSsolution}) through either a sample-splitting conformal construction or a data-augmented conformal construction to have the coverage validity without distributional assumptions.
\subsection{BCOPS}
\label{sec:methods_offline}
 With the observations from the out-of-sample data, we can consider directly problem $\mathcal{P}_k$ with $\mu(x) = f_{test}(x)$:
\begin{align}
\min &\int \mathbbm{1}_{x\in A_k} f_{test}(x)d x\nonumber\\
s.t.\;\;& P_{k}(x\in A_k) \geq 1-\alpha\nonumber
\end{align}  
This has the solution
\[
A_k = \{x: v_k(x) \geq Q(\alpha; v_k, F_k)\}, \;\; v_{k}(x) = \frac{f_k(x)}{f_k(x)+f_{test}(x)}
\]
where we have applied an order-preserving transformation to the density ratio $\frac{f_k(x)}{f_{test}(x)}$ to get $v_k(x)$. Thus, for example, a test point $x$ will have an empty prediction set and be deemed an outlier if each class density $f_k(x)$ relative to the overall density $f_{test}(x)$ is low. 

If we knew $f_k(x)$, $f_{test}(x)$ and hence $v_k(x)$, we can have the oracle $A_k$ and $C(x)$. They are of course unknown: one could use the density estimation to approximate them, but this would suffer in high dimension. Instead, our proposed BCOPS constructs sets $\widehat A_k$ to approximate the above $A_k$ using the idea of the conformal prediction where the conformal score function is learned via  a supervised binary classifier $\mathcal{L}$.  When the density ratio $v_k(x)$ has low dimensional structure, the learned density ratio function from the binary classifier is often much better than that uses the density estimations directly. Since we have used conformal construction when constructing $\widehat A_k$, the constructed prediction set $\widehat C(x) = \{k: x\in \widehat A_k\}$ will also have the finite sample coverage validity. Algorithm \ref{alg:offline} gives details of its implementation. 

\begin{algorithm}
\caption{BCOPS}
\label{alg:offline}
\SetKwInOut{Input}{Input}
\SetKwInOut{Output}{Output}
\underline{function BCOPS}($D^{tr}$, $D^{te}$, $\alpha$, $\mathcal{L}$)\;

\Input{Coverage level $\alpha$, a binary classifier $\mathcal{L}$, labeled training data $D^{tr}$, unlabeled test data $D^{te}$.}
\Output{For each $x\in D^{te}$, the prediction set $\widehat C(x)$.}
\begin{enumerate}
\item Randomly split the training and test data into  $\{D^{tr}_1, D^{tr}_2\}$  and $\{D^{te}_1, D^{te}_2\}$. Let $D^{tr}_{k,1}$, $D^{tr}_{k,2}$ contain samples from class $k$ in $D^{tr}_1$, $D^{tr}_2$ respectively

\item For each $k$, apply  $\mathcal{L}$ to $\{D^{tr}_{k,1}, D^{te}_1\}$ to separate $D^{tr}_{k,1}$ from  $D^{te}_1$ and learn a prediction function $\hat v_{k,1}(x)$ for $v_{k}(x) = \frac{f_k(x)}{f_k(x)+f_{test}(x)}$. Do the same thing with  $\{D^{tr}_{k,2}, D^{te}_2\}$, and denote the learned  prediction function by $\hat v_{k,2}(x)$. 

\item For $x\in D^{te}$, letting $t$ be $\in \{1,2\}$ such that $x\in D^{te}_t$, and $t' = \{1,2\}\setminus t$. We construct
\[
s_k(x) = \frac{1}{|D^{tr}_{k,t}|+1}\sum_{z\in D^{tr}_{k,t}\cup\{x\}}\mathbbm{1}_{\hat v_{k,t'}(x) \geq \hat v_{k,t'}(z)}
\]
and $\widehat A_k = \{x:s_k(x) \geq \frac{\lfloor(|D^{tr}_{k,t}|+1)\alpha\rfloor}{|D^{tr}_{k,t}|+1}\}$, $\widehat C(x) = \{k: x\in \widehat A_k\}$.
\end{enumerate}
\end{algorithm}
\begin{remark}
Algorithm \ref{alg:offline} uses the sample-splitting conformal construction. We can also use the data augmentation conformal prediction instead. For a new observation $x$ and class $k$, we can consider the augmented data $D_k=\{x_{k,1},\ldots, x_{k, n_k}, x\}$,  where $x_{k,i}$ for $i=1,\ldots, n_k$ are samples from class $k$ in the training data set, and  build a classifier separating $D_k$ from $D^{te}\setminus\{x\}$, the test data excluding $x$. For each new observation, we let the trained prediction model be $\hat v_k(.|x)$ and let 
\[
s_k(x) = \frac{1}{n_k+1}\sum_{z\in D^{tr}_{k}}\mathbbm{1}_{\hat v_{k}(x|x) \geq \hat v_{k}(z|x)},\;\widehat A_k = \{x:s_k(x) \geq \frac{\lfloor(n_k+1)\alpha\rfloor}{n_k+1}\}, \;\widehat C(x) = \{k| x\in \widehat A_k\}
\]
By exchangeability, we can also have finite sample coverage guarantee using this approach (data augmentation conformal construction). However, we use the sample-splitting conformal construction in this paper to avoid  a huge computational cost. 
\end{remark}
By exchangeability, we know that the above procedure has finite sample validity \citep{vovk2009line, cadre2009clustering, lei2013distribution, lei2014classification, lei2014distribution}:

\begin{proposition}
\label{prop:prop1}
$\widehat C(x)$ is finite sample valid:
\[
P_{k}(k \in \widehat C(x)) \geq 1-\alpha,\;\forall k = 1,\ldots, K
\]
\end{proposition}
 Algorithm \ref{alg:offline} produces prediction set $\widehat C(x)$ that achieves the same objective as the oracle prediction set $C(x)$ (we will refer it as the asymptotic optimality) if $\hat v_{k,1}$, $\hat v_{k,2}$ are good estimations of $v_k(x)$, and $v_k(x)$ is well-behaved. A more rigorous statement can be found in section \ref{sec:theory}.

\subsection{A simulated example}
\label{sec:methods_example} 
In this section, we provide a simple simulated example to illustrate differences between three different methods: (1) BCOPS, (2) density-level set where $\mu(x) \propto 1$ in $\mathcal{P}$, and (3) in-sample ratio set where $\mu(x) = f(x)$ in $\mathcal{P}$. All three methods have followed the sample-splitting conformal construction with the level $\alpha = 0.05$. For both BCOPS and the in-sample ratio set, we have used the random forest classifier to learn $v_k(x)$ ($v_k(x) = \frac{f_k(x)}{f_k(x)+f_{test}(x)}$ for BCOPS and $v_k(x) = \frac{f_k(x)}{f(x)}$ for the in-sample ratio set).

We let $x\in \real^{10}$ and generate 1000 training samples, half  from class 1 and the other half  from  class 2. The feature $x$ is generated as
\[
x_{1} \sim \left\{\begin{array}{lcr}N(0,1)&\mbox{if}&y = 1 \\ N(3,0.5)&\mbox{if}&y = 2 \end{array}\right.,\;\;x_{j} \sim N(0,1),\; j = 2,\ldots,10
\] 
We have 1500 test samples, one third from class 1, one third  from class 2, and the other one third of them follow the distribution (outliers, class $ \bold{R}$):
\[
x_{2}  \sim N(3,1),\;\;x_{j} \sim N(0,1),\; j \neq 2
\] 
In this example, we let the learning algorithm $\mathcal{L}$ be the random forest classifier. Figure \ref{fig:illustration5} plots the first two dimensions of the test samples and shows the regions with 95\% coverage for BCOPS, density-level set and the in-sample ratio set. The upper left plot colors the data based on its correct label, and is colored black/blue/red if its label is class 1/class 2/outliers. For the remaining three plots,  a sample is colored black if $C(x) = \{1\}$, blue if $C(x) = \{2\}$, green if $C(x) = \{1,2\}$ and red if $C(x) = \emptyset$. Table \ref{tab:illustrationSec2}  shows results of abstention rate in outliers (the higher, the better), prediction accuracy for data from class $1$ and $2$ (a prediction is called correct if $C(x) = \{y\}$ for a sample $(x, y)$), coverages for class 1 and class 2.

\begin{table}[ht]
\centering
\caption{\em An illustrative example. The second column $ \bold{R}$ is the abstention rate of outliers, the third column is the prediction accuracy, the fourth and fifth columns are the coverage for samples from class 1 and class 2. }
\label{tab:illustrationSec2}
\begin{tabular}{lrrrr}
  \hline
 & $ \bold{R}$& accuracy & coverage I & coverage II \\ 
  \hline
density-level & 0.46 & 0.57 & 0.96 & 0.97 \\ 
in-sample ratio & 0.20 & 0.94 & 0.94 & 0.95 \\ 
BCOPS & 0.84 & 0.95 & 0.96 & 0.97 \\ 
   \hline
\end{tabular}
\end{table}
 We can see that the BCOPS achieves much higher abstention rate in outliers, and much higher accuracy in the observed classes compared with the density-level set, while the in-sample ratio set has similar accuracy as the BCOPS but the lowest abstention rate in this example.  
 
 We also observe that small $\alpha$ might lead to making predictions on many outliers which are far from the training data (especially for the density-level set and the in-sample ratio set). While the power for outlier detection varies for different approaches and problems, we want to learn about the outlier abstention rate no matter what kind of method we are using. In section \ref{sec:methods2}, we provide methods for this purpose.
\begin{figure}
\caption{\em A simulated example. The upper left plot shows the class label for each sample in the test data set. The upper right, lower left, lower right plots corresponds to the prediction results using the density-level set, in-sample ratio set and BCOPS respectively.  The upper left plot colors the data based on its correct label, and is colored black/blue/red if its label is class 1/class 2/outliers. For the remaining three plots,  a sample is colored black if $C(x) = \{1\}$, blue if $C(x) = \{2\}$, green if $C(x) = \{1,2\}$ and red if $C(x) = \emptyset$.}
\label{fig:illustration5}
\begin{center}
\includegraphics[width = 1\textwidth, height = 1\textwidth]{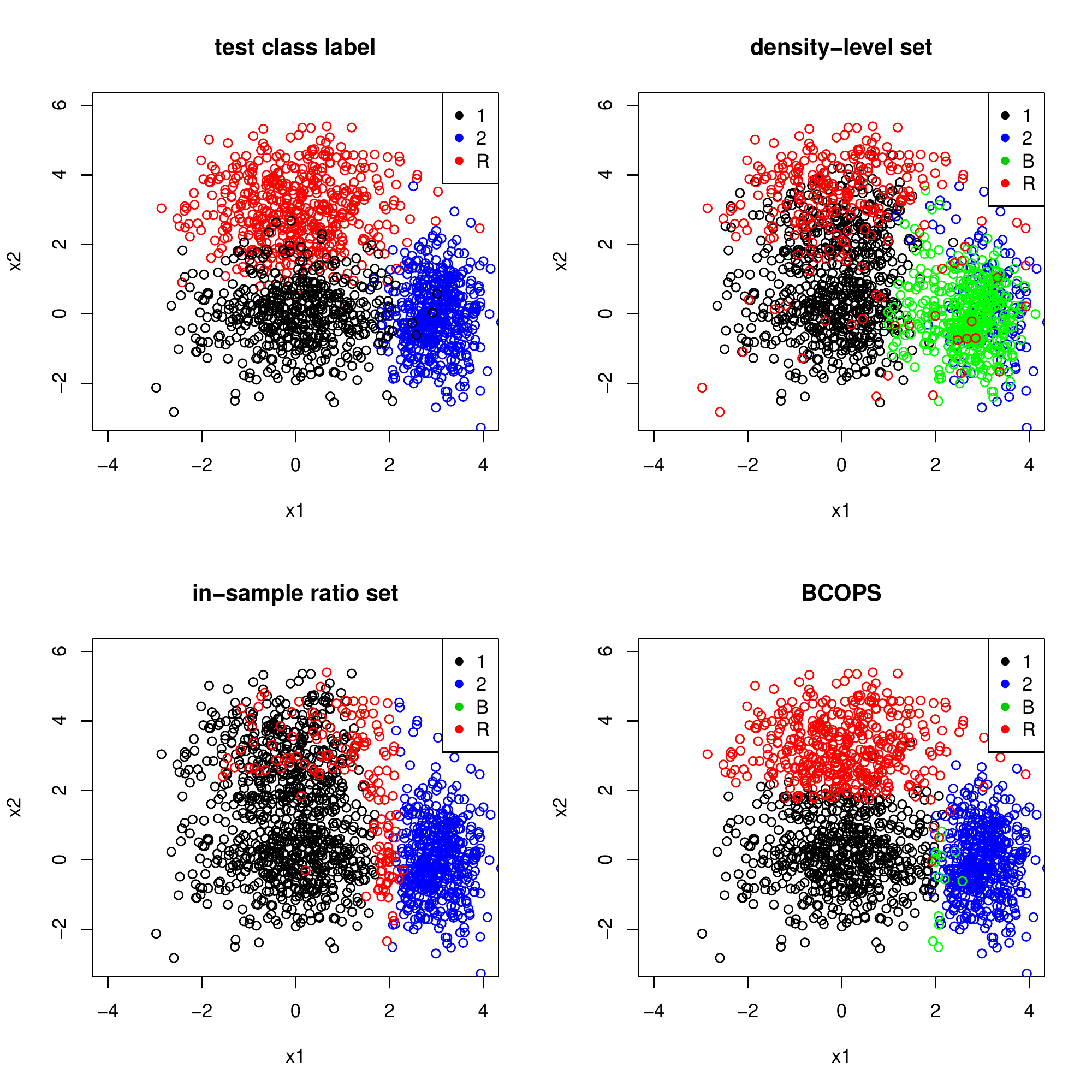}
\end{center}
\end{figure}

\section{Outlier abstention rate and false labeling rate}
\label{sec:methods2}
In the previous section, we proposed BCOPS for prediction set construction at a given level $\alpha$. In this section, we describe a regression-based method to estimate the test set mixture proportions $\tilde \pi_k$ for $k=1,\ldots, K$, and using this, we estimate the outlier abstention rate and FLR (false labeling rate). The outlier abstention rate is the expected proportion of outliers with an empty prediction set. The FLR is the expected ratio between the number of outlier given a label and total number of samples given a label.  For a fixed prediction set function $C(x)$, its outlier abstention rate $\gamma$ and FLR are defined as
\begin{align*}
\gamma &\coloneqq P_{\bold{R}}(C(x) = \emptyset)\\
\text{FLR} &\coloneqq E\left(\frac{|\{x\in D^{te}: y(x) = \bold{R}, C(x)\neq \emptyset\}|}{|\{x\in D^{te}: C(x)\neq \emptyset\}|\vee 1}\right)
\end{align*}
The expectation is taken over the distribution of $x$.  The outlier abstention rate is power for $C(x)$ in terms of the outlier detection while FLR controls the percent of outliers among samples with predictions. 

Information about the outlier abstention rate and FLR can be valuable when picking $\alpha$. For example,  while we want to have both small $\alpha$ and large $\gamma$,  $\gamma$ is negatively related to $\alpha$. As a result, we might want to choose $\alpha$ based on the tradeoff curve of $\alpha$ and $\gamma$.   There are different ways that we may want to utilize $\gamma$ or FLR:
\begin{enumerate}
\item Set $\alpha \geq \alpha^*$ to control the FLR, for example, let $\alpha^* = \inf \{\alpha: \text{FLR}(\alpha) \leq 10\%\}$ where $\text{FLR}(\alpha)$ is the FLR at the given $\alpha$.
\item Set $\alpha = \alpha^*$ to control the abstention rate where $\alpha^*$ is the smallest $\alpha$ such that $\gamma$ is above a given threshold.

\item Set $\alpha$ without considering FLR or $\gamma$, however, in this case, we can still assign a score for each point to measure how likely it may be an outlier. For each point $x$, let $\alpha(x)$ be the largest $\alpha$ such that $C(x) = \emptyset$. We let its outlier score be $\gamma(x) \coloneqq \gamma_{\alpha(x)}$, where $\gamma_{\alpha(x)}$ is the abstention rate at the required coverage is $(1-\alpha(x))$. The interpretation for the outlier score is simple: if we want to refrain from making prediction for $\gamma$ proportion of outliers, then we do not make a prediction for samples with $\gamma(x) \geq \gamma$ even if $C(x)$ itself is non-empty. 
\end{enumerate}

\subsection{Estimation of $\gamma$ and FLR}
\label{sec:methods2-1}
When the proportion of outliers $\epsilon$ is  greater than zero in the test samples, $\gamma$ can also be expressed as
\[
\gamma = \frac{E[\mbox{Number of outliers with }C(x) = \emptyset]}{\mbox{Total number of outliers}} = \frac{E[N_{\emptyset}] - \sum^N_{k=1} N\tilde\pi_k \gamma_k}{N(1-\sum^K_{k=1}\tilde\pi_k)}
\]
where $N$ is the total number of test samples, $N_{\emptyset}$ is the total number of samples with abstention ($C(x) = \emptyset$) and $\gamma_k \coloneqq P_k(C(x) = \emptyset)$ is the abstention rate for class $k$.  The FLR can be expressed as
\[
\text{FLR}=E[\frac{\overbrace{(N-N_{\emptyset})}^{\mbox{All non-empty}}-\sum^K_{k=1}\overbrace{N\tilde \pi_k(1- \gamma_k)}^{\mbox{Class $k$ non-empty}}}{(N-N_{\emptyset})\vee 1}]
\]

When $\hat \gamma_k$ and $\hat \pi_k$, the estimates of $\gamma_k$ and $\tilde\pi_k$, are available, we can construct empirical estimates $\hat \gamma$ and $\widehat{\text{FLR}}$ of $\gamma$ and FLR:
\[
\hat \gamma = \frac{(N_{\emptyset}-\sum^K_{k=1}N\hat \pi_k \hat \gamma_k)\vee 0}{N(1-\sum^K_{k=1}\hat \pi_k)\vee 1},\;\widehat {\text{FLR}}= \frac{(N-N_{\emptyset}-\sum_k N\hat \pi_k(1-\hat \gamma_k))\vee 0}{(N-N_{\emptyset})\vee 1}
\]
\subsubsection{Estimation of $\gamma_k$}
We estimate $\gamma_k$ using the empirical distribution for class $k$ from the training data. More specifically, for BCOPS, we let
\[
\hat \gamma_k= \frac{|\{x\in D^{tr}_k: \widehat C(x) = \emptyset\}|}{|D^{tr}_k|}
\]
where $\widehat C(x)$ follows the same construction as in the BCOPS Algorithm: For $x\in D^{tr}_{k,t}$, we construct $\hat C(x)$ for $x$ using training and test samples from fold $t' \in \{1,2\}\setminus t$ as described in the BCOPS Algorithm.

\subsubsection{Estimation of $\tilde \pi_k$}
\label{sec:methods_shift}
Let $S_l$ be regions such that $P_{\bold{R}}(S_l) = 0$ for $l = 1,\ldots, K$. Then, by our model assumption, we know
\[
P_{test}(S_l) = \sum^K_{k=1}\tilde \pi_k P_k(S_l)
\]
Let ${\rP}_l = P_{test}(S_l)$ be the response vector, $\bold{P}$ be a $K\times K$ design matrix with $\bold{P}_{l,k}=P_k(S_l)$, and $\bold\Sigma = \bold{P}^T\bold{P}$. As long as $\bold\Sigma$ is invertible, $\tilde \pi$ is the solution to the regression problem that regresses $\rP$ on $\bold{P}$.  Next, we give a simple proposal trying to construct such $S_l$.

For a fixed function $\eta:\real^p\rightarrow \real^K$, let $g_{l,k}(.)$ be the density of $\eta_l$ in class $k$. If the outliers happen with probability 0 at regions of $\eta_l$ where class $l$ has relatively high density, we can let $S_l$ be the region with relatively high $g_{l,l}(.)$: 
\begin{enumerate}
\item Let $\circ$ be the composition operator and $S_l = \{z: g_{l,l}(z) \geq Q(\zeta; g_{l,l}\circ \eta_l, F_l) \}$,  $P_{l} = P_{test}(\eta_l(x)\in S_l)$,   $\bold P_{l,k} = P_{k}(\eta_l(x)\in S_l)$ for a user-specific constant $\zeta\in (0,1)$ specifying the separation between inliers and outliers.
\item We would like to solve $J(\eta) \coloneqq \min_{\pi}\| \rP- \bold{P}\pi\|^2_2$, the the oracle  problem based on the function $\eta$.
\end{enumerate}
We recommend taking $\eta_l(x)=\log \frac{f_{l}(x)}{f_{test}(x)}$, the log-odd ratio separating class $l$ from the test data, since it automatically tries to separate class $l$ from other classes, including the outliers. 

Neither $\eta$ nor $\rP, \bP$ given $\eta$ will be observed. In practice, we use sample-splitting to estimate $\eta$ in one fold of the data and estimate $\rP$, $\bP$ empirically  in the other fold conditional on the estimated $\eta$. See Algorithm \ref{alg:shift}, MixEstimate (mixture proportion estimation), for details. 
\begin{algorithm}[H]
\caption{\em MixEstimate}
\label{alg:shift}
\SetKwInOut{Input}{Input}
\SetKwInOut{Output}{Output}
\underline{function MixEstimate}($D^{tr}$, $D^{te}$, $\zeta$, $\mathcal{L}$)\;

\Input{Left-out proportion $\zeta$, a binary classifier $\mathcal{L}$, labeled training data $D^{tr}$, unlabeled test data $D^{te}$. By default, $\zeta  = 0.2$.}
\Output{Estimated mixture proportion $\{\hat \pi_k,\; k = 1,\ldots, K\}$}
\begin{enumerate}
\item  Randomly split the training and test data into $\{D^{tr}_1, D^{tr}_2\}$  and $\{D^{te}_1, D^{te}_2\}$. For $t = 1,2$, let $D^{tr}_{k,t}$ contain samples from class $k$ in $D^{tr}_t$ , and apply $\mathcal{L}$ to $\{D^{tr}_{k,t}, D^{te}_t\}$ to separate samples from class $k$  and the test set,  we get $\hat \eta^t_{l}(x)$ as the estimate to  $\eta_l(x) = \log \frac{f_l(x)}{f_{test}(x)}$.
\item For fold $t = 1,2$: let $t' = \{1,2\}\setminus t$, and
\begin{itemize}
\item let $\hat F^{t'}_l$ and $\hat g^{t'}_{l,l}(.)$ be empirical distribution of class $l$ and the gaussian kernel density estimation of the density of $\hat \eta^t_l(x)$ using $D^{tr}_{l,{t'}}$. 
\item the empirical problem $\hat J(\hat \eta^t)$ is constructed with empirical probabilities of each class in fold $t'$ falling into regions $\widehat S_l =\{t: \hat g^{t'}_{l,l}(t) \geq Q(\zeta; \hat g^{t'}_{l,l}\circ \hat \eta^t, \hat F^{t'}_l)\}$. Let $\hat \pi^t_k$ for $k = 1,\ldots, K$ be the solutions to $\hat J(\hat \eta^t)$.
\end{itemize}
\item Output the average mixture proportion estimate and let  $\hat \pi_{k} = \frac{\hat \pi_{k,1}+\hat \pi_{k,2}}{2}, \;\forall k = 1,\ldots, K$.
\end{enumerate}
\end{algorithm}
\begin{remark}
\begin{itemize}
\item  Our proposal can be viewed as an extension to the BBSE method in \cite{lipton2018detecting} under the presence with outliers.
\item In practice, we can add a constraint on the optimization variables $\pi_k$ and require that  $\sum^K_{k=1} \pi_k \leq 1$ and  $\pi_k\geq 0$.
This constraint guarantees that both $\tilde \pi_k$ and the outlier proportion $\epsilon$ are non-negative.
\end{itemize}
\end{remark}

We show in section \ref{sec:theory} that the estimates  from MixEstimate will converge to $\tilde \pi_k$ under proper assumptions. As a continuation of the example shown in section \ref{sec:methods_example}. Figure \ref{fig:illustration6} shows curves of estimated FLR, estimated outlier abstention rate $\hat \gamma$, as well as the FLP (false labeling proportion), which is the sample version of FLR using current test data, and the sample version of the $\gamma$ against different $\alpha$. 
\begin{figure}[H]
\caption{\em An illustrative example, continuation of  section \ref{sec:methods_example}. The red solid/dashed curves show the FLP and estimated FLR against different $\alpha$. The blue solid/dashed curves show the actual outlier abstention rate $\gamma$ and estimated $\gamma$ against different $\alpha$.}
\label{fig:illustration6}
\begin{center}
\includegraphics[width = .6\textwidth, height = .4\textwidth]{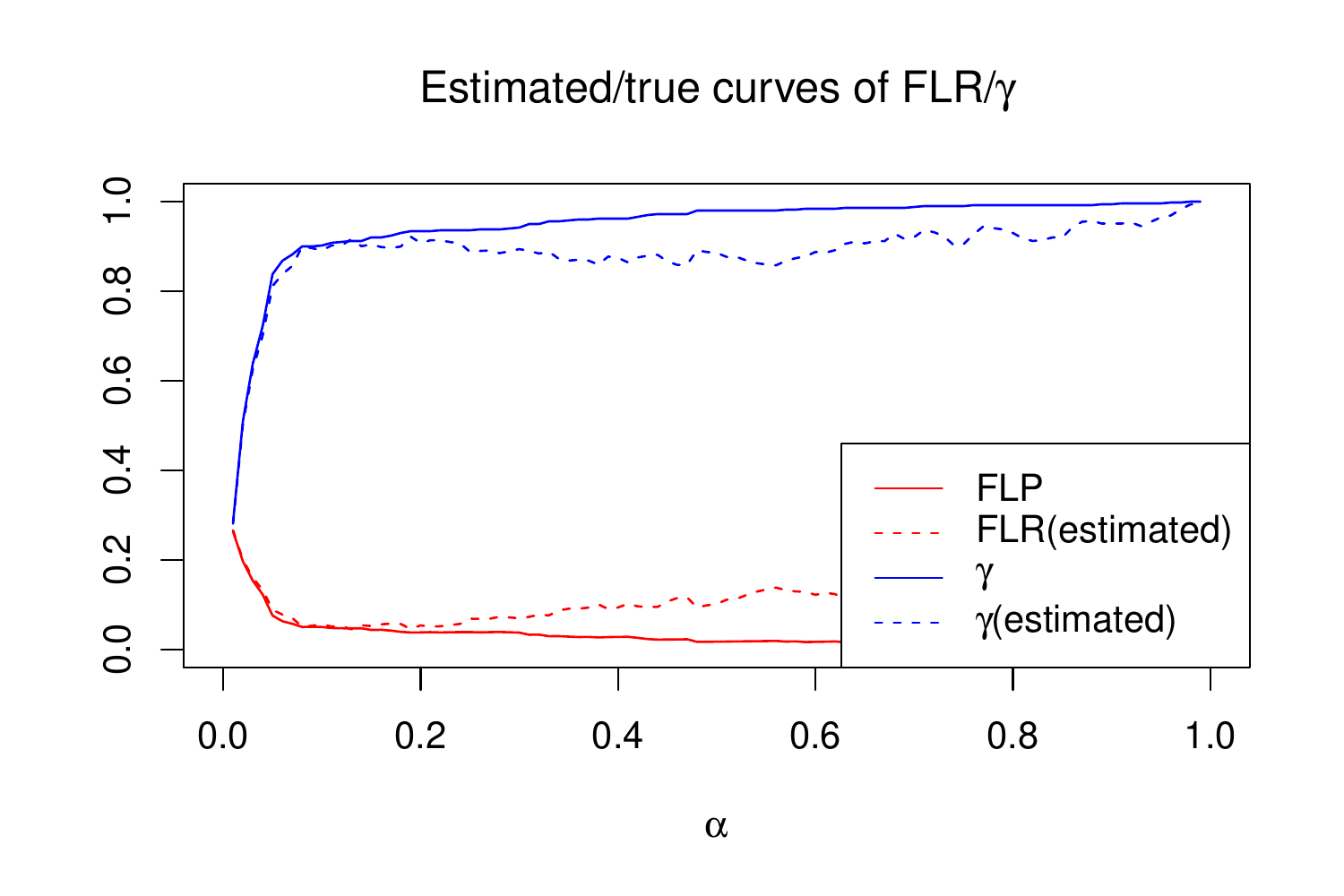}
\end{center}
\end{figure}

\section{Asymptotic properties of BCOPS and MixEstimate}
\label{sec:theory}
Let $n_k$ be the sample size of class $k$ in the training data and $n$ be the size of the training data. Let $N$ be the size of the test data. In this section, we consider the asymptotic regime where $n\rightarrow\infty$, and assume that  $\lim_{n\rightarrow\infty}\frac{N}{n} \geq c$ for a constant $c > 0$ and $\frac{n_k}{n}\rightarrow c_k$ for a constant $c_k \in (0,1)$, $k = 1,\ldots, K$. In this asymptotic regime, we show that 
\begin{enumerate}
\item The  prediction set $\widehat C(x)$ constructed using BCOPS achieves the same loss as the oracle prediction set $C(x)$  asymptotically if the estimation of $v_k(x) = \frac{f_k(x)}{f_k(x)+f_{test}(x)}$ is close to it, and under some conditions on the densities of $x$ and distribution of $v_k(x)$ for $k =1,\ldots, K$.
\item The mixture proportion estimations converge to the the true out-of-sample mixture proportions $\tilde \pi_k$ if the outliers are rare when the observed classes have high densities, and under some conditions on the densities of $\eta_l(x)$, the functions used to construct $S_l$.
\end{enumerate}

\subsection{Asymptotic optimality of the  BCOPS}
Let $\hat v_k(x)$ be the estimate of $v_k(x)$, representing either $\hat v_{k,1}(x)$ or $\hat v_{k,2}(x)$ in the  BCOPS Algorithm.
\begin{assumptions}
\label{ass:ass1}
Densities $f_1(x), \ldots, f_k(x), e(x)$ are upper bounded by a constant. There exist constants $0 < c_1 \leq c_2$ and $\delta_0,\;\gamma > 0$, such that for $k = 1,\ldots, K$,  we have
\[
c_1|\delta|^\gamma \leq |P_k(\{x| v_k(x) \leq Q(\alpha; v_k, F_k)+\delta\}) - \alpha| \leq c_2 |\delta|^\gamma, \forall -\delta_0\leq \delta \leq \delta_0
\]
\end{assumptions}
\begin{remark}
We require that the underlying function $v_k(x)$ is neither too steep nor too flat around the boundary of the optimal decision region $A_k$. This makes sure that this boundary is not too sensitive to small errors in estimating $v_k(x)$ and the final loss is not too sensitive to small changes in the decision region.
\end{remark}
\begin{assumptions}
\label{ass:ass2}
The estimated function $\hat v_k(x)$ converges to the true model $v_k(x)$: there exists constants $B$, $\beta_1, \beta_2> 0$ and a set $A_n$ of $x$ depending on $n$, such that,  as $n\rightarrow \infty$, we have  $P(\sup_{x\in A_n}|\hat v_{k}(x) - v_{k}(x)|< B(\frac{\log n}{n})^{\frac{\beta_1}{2}})\rightarrow 1$, $P_{test}(x\in A_n)\geq1-B(\frac{\log n}{n})^{\frac{\beta_2}{2}}$.
\end{assumptions}
\begin{remark}
For such an assumption to hold in high dimensional setting, the classifier $\mathcal{L}$ in BCOPS usually needs to be parametric and we will also need some parametric model assumptions depending on $\mathcal{L}$. For example, when we let $\mathcal{L}$ be the logistic regression with lasso penalty in BCOPS, we could require the approximate correctness of the logistic model, nice behavior of features and sparsity in signals \citep{van2008high}.
\end{remark}
\begin{theorem}
\label{thm:thm1}
Under Assumptions \ref{ass:ass1}-\ref{ass:ass2}, for any fixed level $\alpha > 0$, let  $C(x)$ be the oracle BCOPS prediction set, for a large enough constant $B$, we have
\[
P(\int (|\widehat C(x)|-|C(x)|)f_{test}(x) dx \geq B(\frac{\log n}{n})^{\frac{\min(\gamma\beta_1,\beta_2,1)}{2})}) \rightarrow 0
\]
\end{theorem}
Proof of Theorem \ref{thm:thm1} is given in Appendix \ref{app:BCOPSproof}.

\subsection{Asymptotic consistency of  the  mixture proportion estimates}
\label{sec:theory_shift}
In this section,  let $\hat \eta$ represent $\hat \eta^t$ for $t = 1,2$ in label shift estimation Algorithm (Algorithm \ref{alg:shift}).   As defined in section \ref{sec:methods_shift}: $\zeta$ is a user-specific positive constant, $\eta: \real^p\rightarrow \real^K$ is a fixed function, and $g_{l,k}(t)$ is  the density of $\eta_l$ in class $k$, $S_l= \{t:g_{l,l}(t) \geq Q(\zeta; g_{l,l}\circ \eta_l, F_l)\}$ and $\rP$, $\bP$ are the oracle response and design matrix based on $\eta$, $\bSigma = \bP^T\bP$. We also let  $\hat J(\eta)$ be the problem with empirical $\rP$ and $\bP$ from fold $t = 1$ or $2$, and $h_n$ be the bandwidth of the gaussian kernel density estimation in  Algorithm \ref{alg:shift}. The bandwidth $h_n$ satisfies $h_n\rightarrow 0$ and $\frac{\log n}{nh_n}\rightarrow 0$.


\begin{assumptions}
\label{ass:ass3}
For $k = 1,\ldots, K, \bold{R}$ and $l=1,\ldots, K$, the density  $g_{l,k}(t)$ is bounded, and $g_{l,l}(t)$ is $H\ddot{o}lder$ continuous (e.g. there exist constants $ 0<\gamma\leq 1$ and $B$, such that $|g_{l,l}(z) - g_{l,l}(z')|\leq B|z-z'|^{\gamma}$ for $\forall z, z'\in \real$).
\end{assumptions}
\begin{assumptions}
\label{ass:ass4}
There exist constants $\gamma,\; c_1, \;c_2 > 0$ and $\delta_0 > 0$, such that for $\forall l = 1,\ldots,K$:
\[
c_1 |\delta|^{\gamma}\leq  |P_{l}(g_{l,l}(t)\leq Q(\zeta; g_{l,l}\circ \eta_l, F_l)+\delta)-\zeta|\leq c_2|\delta|^{\gamma} ,  \forall -\delta_0\leq \delta \leq \delta_0
\]
\end{assumptions}
\begin{remark}
Assumption \ref{ass:ass4} is similar to Assumption \ref{ass:ass1}, it asks that $g_{l,l}(t)$ is neither too steep nor too flat around the boundary of $S_l$.
\end{remark}
\begin{assumptions}
\label{ass:ass5}
$P_{\bold{R}}(\eta_l(x)\in S_l) = 0,\;\forall l = 1,\ldots,K$, and $\bSigma$ is invertible  with the smallest eigenvalue $\sigma_{\min} \geq c$ for a constant $c > 0$.
\end{assumptions}


\begin{theorem}
\label{thm:shift}
Under Assumption \ref{ass:ass3}-\ref{ass:ass5}, let $\{\hat \pi_k, k = 1,\ldots,K\}$ be solutions to $\hat J(\eta)$. Then $\hat \pi_k\overset{p}{\rightarrow} \tilde \pi_k$  as $n\rightarrow\infty$.
\end{theorem}

By the independence of the two folds, once we have learned $\hat \eta$ in Algorithm \ref{alg:shift}, we can condition on it and treat it as fixed. Hence, we have Corollary \ref{col:col1} as a direct application of Theorem.
\begin{corollary}
\label{col:col1}
Let  $\mathcal{A}_0$ be the event that Assumptions \ref{ass:ass3}-\ref{ass:ass5} are satisfied for $\eta = \hat \eta$. If $P(\mathcal{A}_0)\rightarrow 1$ as $n\rightarrow \infty$. Then, let $\{\hat \pi_k, k = 1,\ldots,K\}$ be solutions to $\hat J(\hat \eta)$, we have $\hat \pi_k\overset{p}{\rightarrow} \tilde \pi_k$  as $n\rightarrow\infty$.
\end{corollary}
Proof of Theorem \ref{thm:shift} is given in Appendix \ref{app:BCOPSproof}.

\section{Real Data Examples}
\label{sec:realData}
\subsection{MNIST handwritten digit example}
We look at the MNIST handwritten digit data set\citep{lecun-mnisthandwrittendigit-2010}. We let the training data contain digits 0-5, while the test data contain digit 0-9. We subsample 10000 training data and 10000 test data and compare the type I and type II errors using different methods. The type I error is defined as ($1-$coverage) and no type I error is defined for new digits  unobserved  in the training data. Rather than considering ${\rm Err}$ defined in eq.(\ref{eq:Err}), we define the type II error as 
${\rm Err} = E_{F}\mathbbm{1}_{|C(x)\setminus \{y\}| > 0}$ for samples with distribution $F$ ($y$ is the true label of $x$), so that we will have $\Err$ in $[0,1]$.  In this part, we consider three methods:
\vskip -.2in
\begin{enumerate}
\item BCOPS: the  BCOPS with the supervised learning algorithm $\mathcal{L}$ being random forest (rf) or logistic+lasso regression (glm).
\item DLS: the density-level set (with $\mu(x) = 1$ in problem $\mathcal{P}$) with the sample-splitting conformal construction.
\item IRS: the in-sample ratio set (with $\mu(x) = f(x)$ in problem $\mathcal{P}$) with the sample-splitting conformal construction and a supervised $K$-class classifier $\mathcal{L}$  to learn $\frac{f_k(x)}{f(x)}$. Here, we also let $\mathcal{L}$ be either random forest (rf) or multinomial+lasso regression (glm). 
\end{enumerate}

Figure \ref{fig:digit1} plots the nominal type I error for digits showing up in the  training data (average), we can see that all methods can control its claimed type I error. 
\begin{figure}
\caption{\em Type I error control: Actual type I error vs nominal type I error $\alpha$ for different methods. All methods can control the targeted type I error.}
\label{fig:digit1}
\begin{center}
\includegraphics[width = .6\textwidth, height = .5\textwidth]{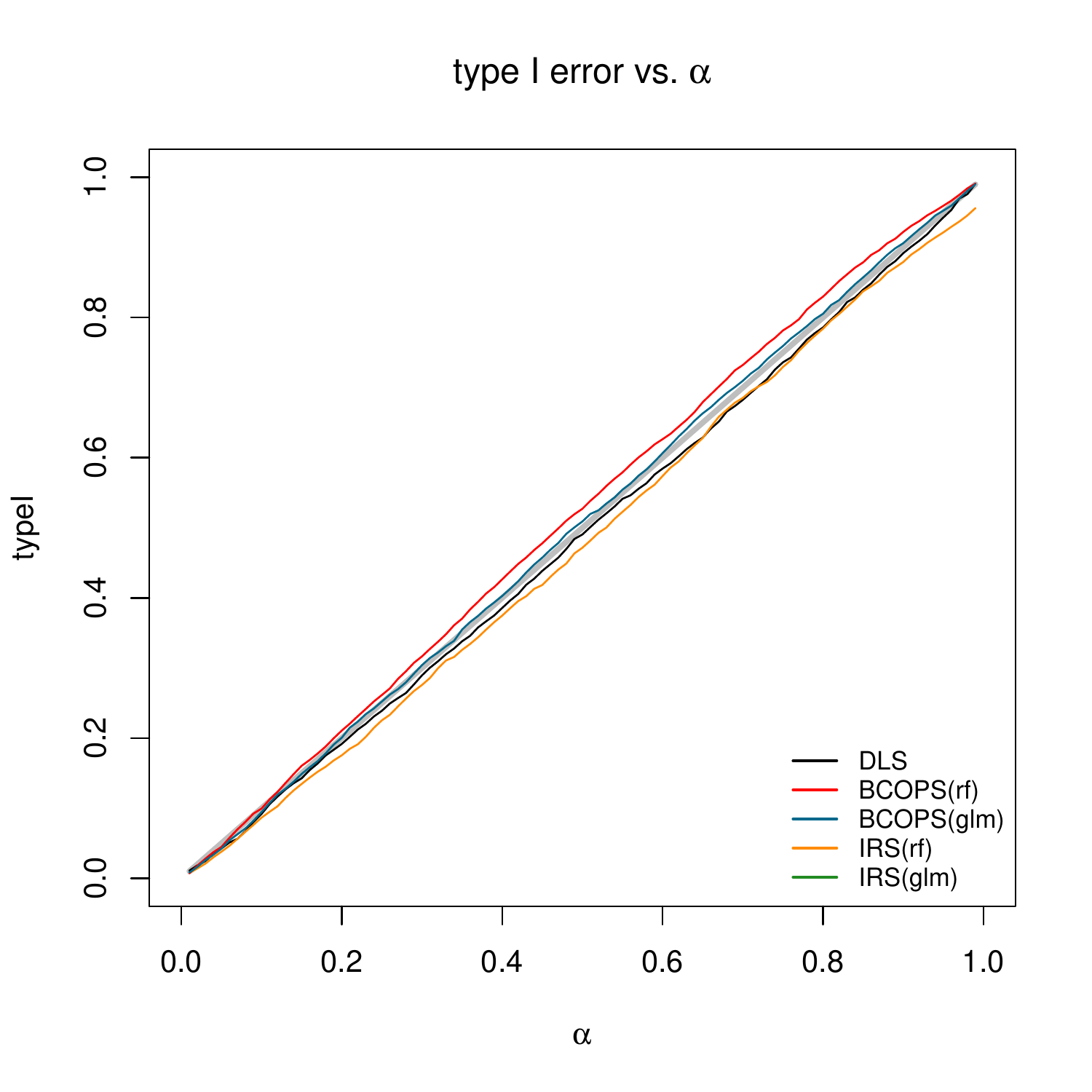}
\end{center}
\end{figure}

\begin{figure}
\caption{\em Comparisons of the Type II $\sim$ Type I error curves using different  conformal prediction sets.  Results for observed digits (digits $\leq 5$) and unobserved digits (digits $\geq 6$) have been presented separately. BCOPS performs the best for the unobserved digits, and IRS is slightly better than BCOPS for the observed digits using the same learning algorithm. Both BCOPS and IRS are much better than DLS in this example. }
\label{fig:digit2}
\begin{center}
\includegraphics[width = 1\textwidth, height = .4\textwidth]{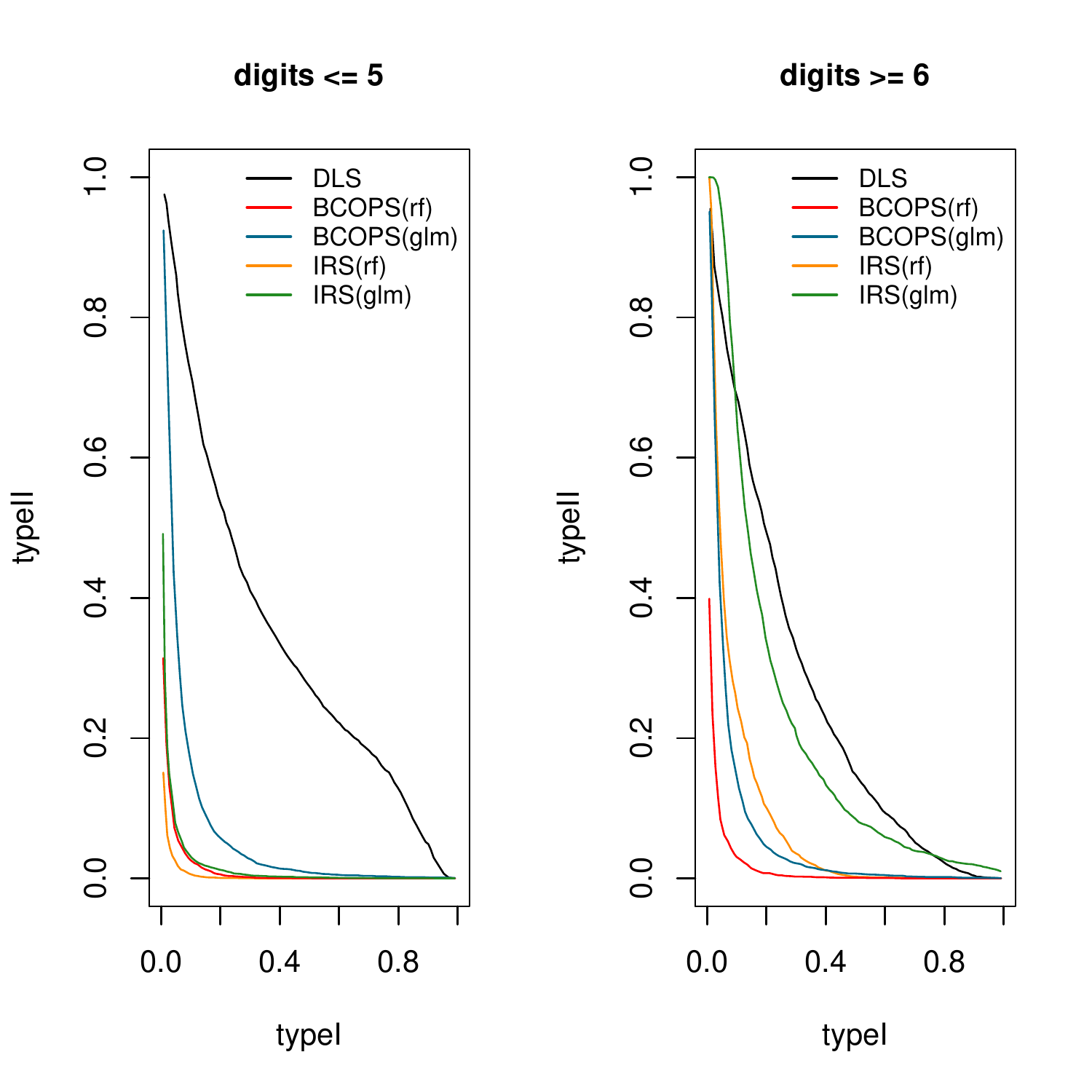}
\end{center}
\end{figure}

Figure \ref{fig:digit2} shows plot of the type II error agains type I error, separately for digits in and not in the training  set,  as $\alpha$ ranges from $0.01$ to $0.99$. We observe that
\begin{itemize}
\item  For the unobserved digits in the training data, we see that

\begin{center}
 DLS$<$IRS(glm)$<$IRS(rf)$<$BCOPS(glm)$<$BCOPS(rf) (ordered from worse to better)
 \end{center}
 
 BCOPS achieves the best performance borrowing information from the unlabeled data. In this example, IRS also has better results compared with DLS for the unobserved digits. IRS depends only on the predictions from the given classifier(s) and does not prevent us from making prediction at a location with sparse observations if the classifiers  themselves do not take this into consideration.  Although we can easily come up with situations where such methods fail entirely in terms of outlier detection, e.g, the simulated example in section \ref{sec:methods_example}, in this example, the dimension learned by IRS for in-sample classification is also informative for outlier detection. 
\item For the observed digits in the training data,  we see that 

\begin{center}
DLS$<$BCOPS(glm) $<$IRS(glm)$<$ BCOPS(rf)$<$IRS(rf)
 \end{center}
 
DLS performs much worse than both BCOPS and IRS, and BCOPS performs  slightly worse than IRS for a given learning algorithm $\mathcal{L}$. 
\end{itemize}
Overall, in this example, BCOPS trades off some in-sample classification accuracy for  higher power in outlier detection.

In practice, we won't have access to the curves in Figure \ref{fig:digit2}. While we can estimate the behavior of the observed digits using the training data, we won't have such luck for the outliers. We can use methods proposed in the section \ref{sec:methods2} to estimate the FLR  and the outlier abstention rate $\gamma$. Figure \ref{fig:digit3} compares the estimated FLR and $\gamma$ with the actual sample-versions of FLR and $\gamma$. We can see that the estimated FLR and $\gamma$ matches reasonably well with the actual FLP and $\gamma$(sample-version) for both learning algorithm $\mathcal{L}$.

\begin{figure}
\caption{\em FLR and outlier abstention rate $\gamma$ estimation. Red curves are actual FLP and estimated FLR, and the blue curves are the actual abstention rate $\gamma$ realized on the current data set and the estimated $\gamma$.}
\label{fig:digit3}
\vskip -.4in
\begin{center}
\includegraphics[width = .49\textwidth, height = .35\textwidth]{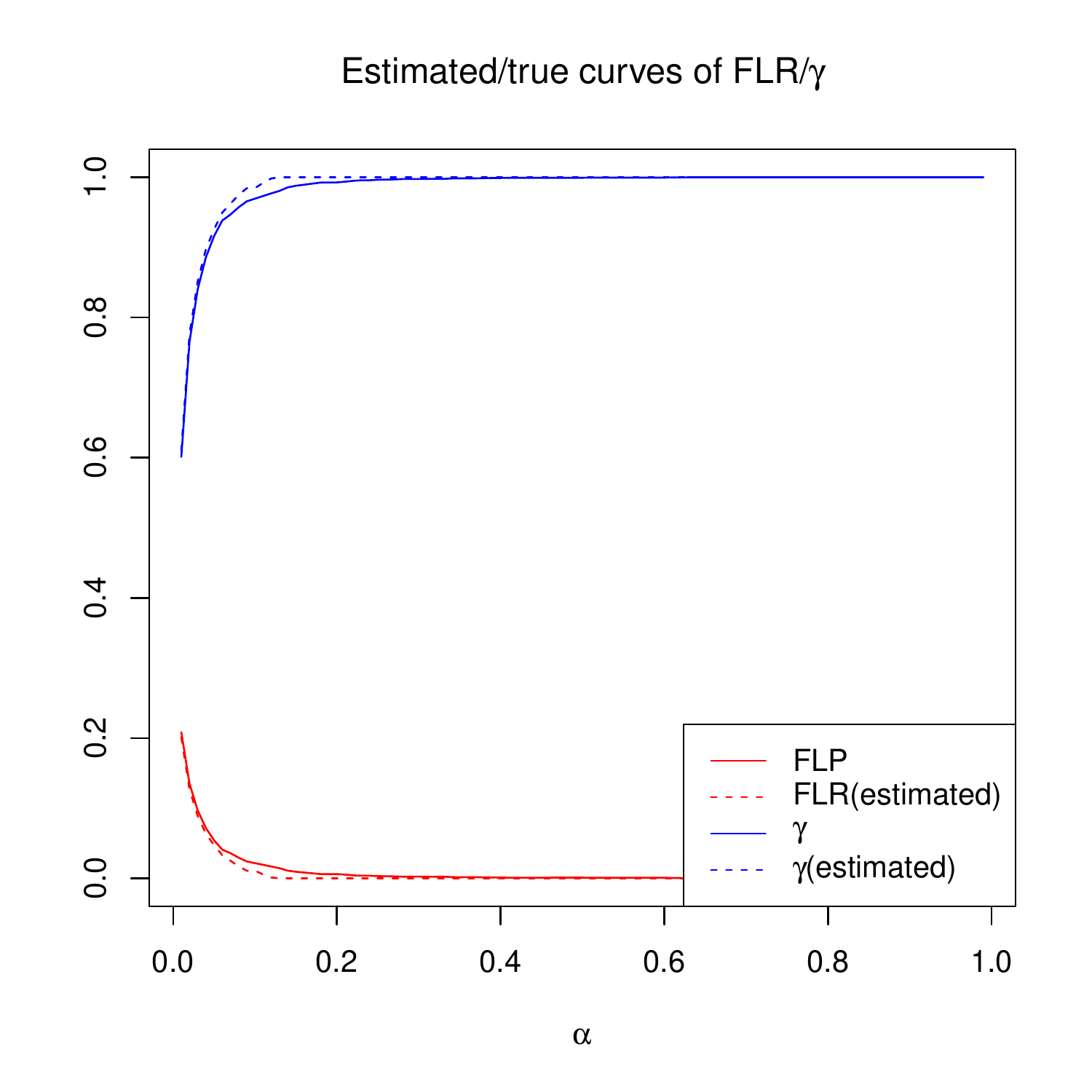}
\includegraphics[width = .49\textwidth, height = .35\textwidth]{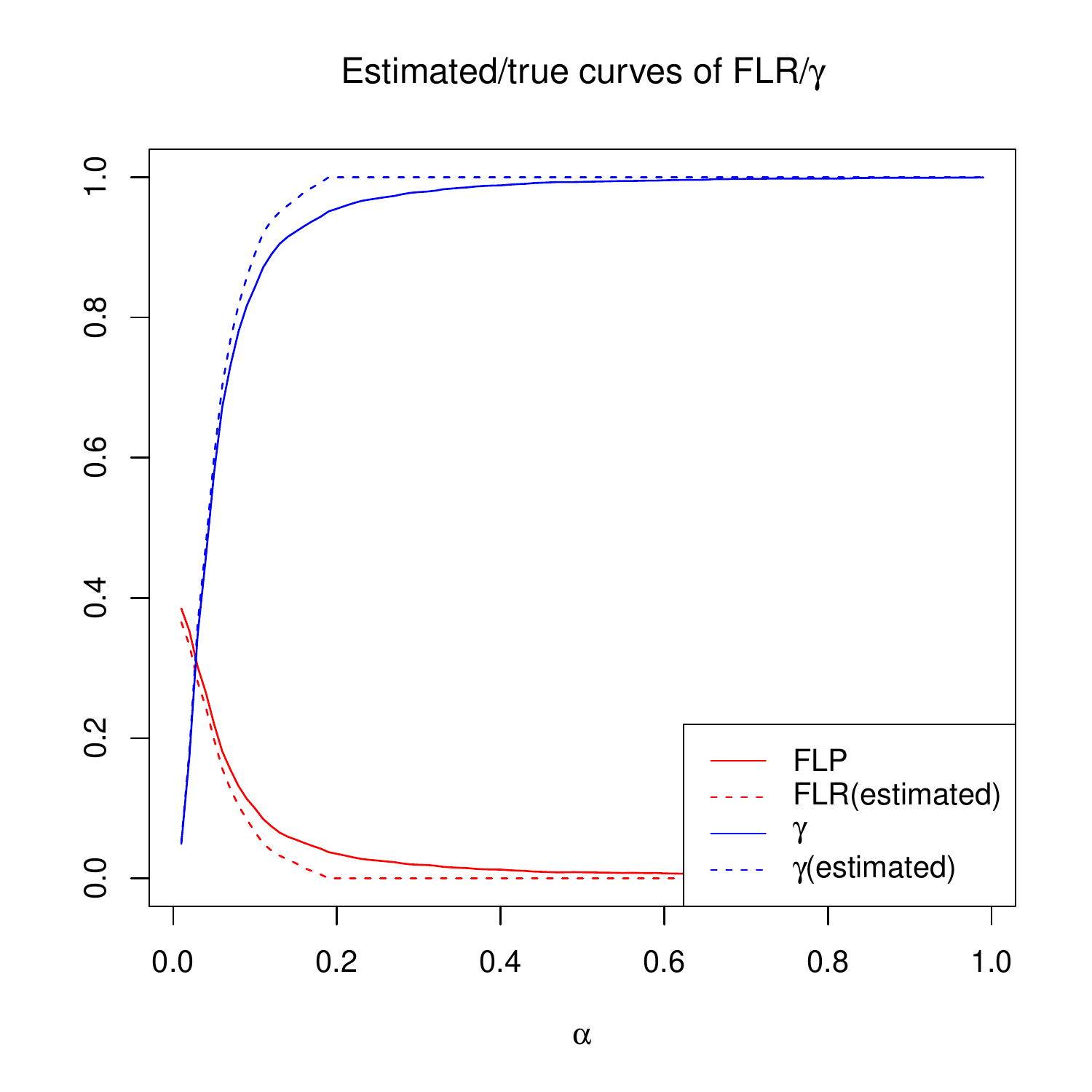}
\end{center}
\end{figure}
\subsection{Network intrusion detection}
The intrusion detector learning task is to build a predictive model capable of distinguishing between ``bad'' connections, called intrusions or attacks, and ``good'' normal connections (\cite{stolfo2000cost}). We use 10\% of the data used in the 1999 KDD intrusion detection contest. We have four classes: normal, smurf,  neptune and other intrusions. The  normal, smurf, neptune samples are randomly assigned into the training and test samples while other intrusions appear only in the test data. We have 116 features in total and approximately 180,000 training samples and  180,000 test samples, and about 3.5\% of the test data are other intrusions. 

We let $\mathcal{L}$ be random forest, and compare the  BCOPS with the RF prediction. Figure \ref{fig:intrusion1} shows the estimated abstention rate and estimated in-sample accuracy defined as (1-estimated type II errors for observed classes). In this example, BCOPS takes $\alpha = 0.05\%$, the largest $\alpha$ achieving 95\% of the abstention rate for outliers.
\begin{figure}[H]
\caption{\em Estimated outlier abstention rate  $\gamma$ and in-sample accuracy against $\alpha$. The vertical red line shows the suggested value for $\alpha$ if we let the estimated abstention rate be $95\%$. In this example, the in-sample accuracy remains almost 1 for extremely small $\alpha$ and is not very instructive for picking $\alpha$. }
\label{fig:intrusion1}
\begin{center}
\includegraphics[width = .49\textwidth, height = .49\textwidth]{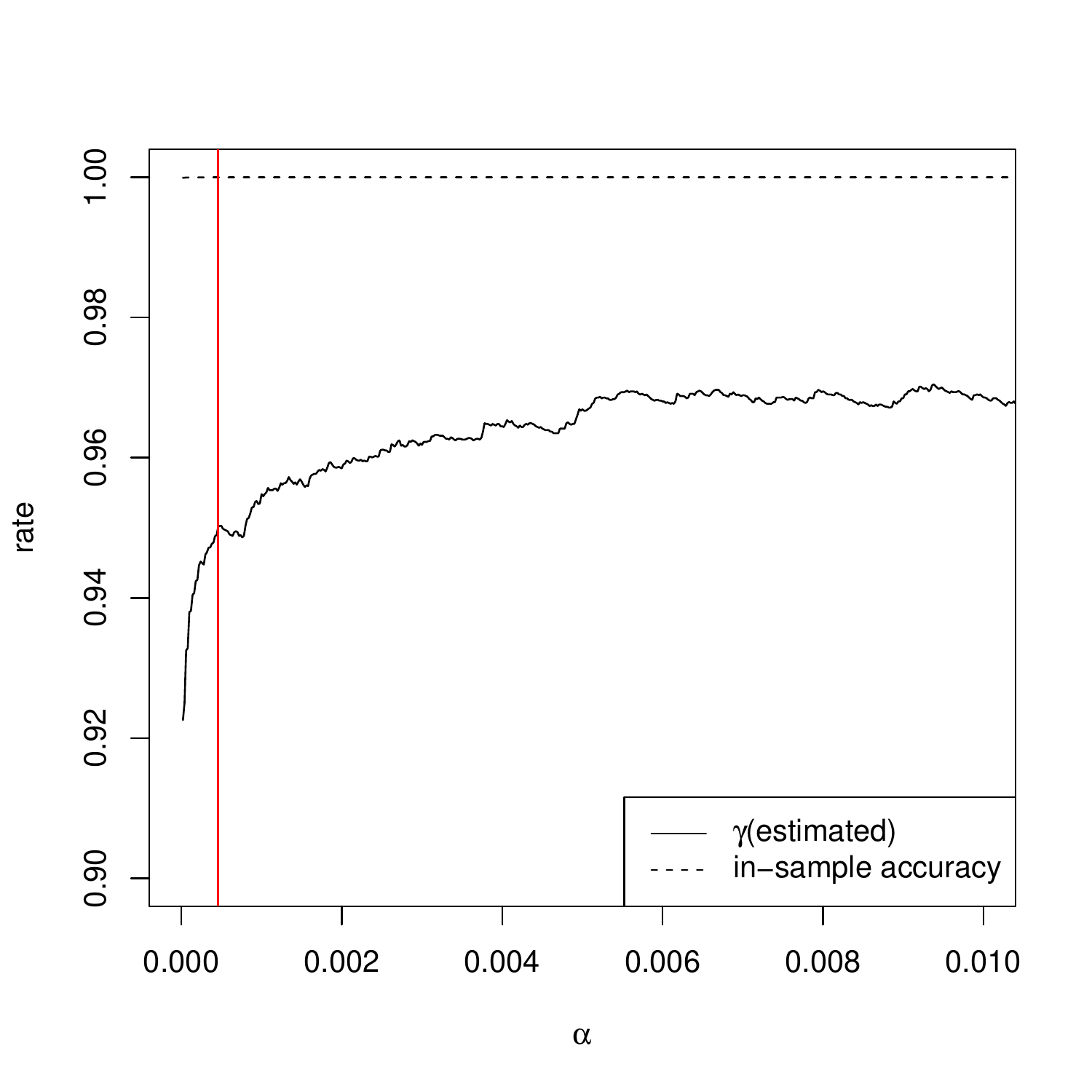}
\end{center}
\end{figure}

Table \ref{tab:intrusion1}  shows prediction accuracy using  BCOPS and RF. We pool smurf and neptune together and call them the observed intrusions. We say that a prediction is correct from RF if it correctly assign normal label to normal data  or assign intrusion label to  intrusions. From  Table \ref{tab:intrusion1},  we can see that the original RF classifier assigns correct label to 99.999\% of samples from the  the observed classes, but claims more 50\% of other intrusions to be the normal. The significant deterioration on unobserved intrusion types is also observed in participants' prediction during the contest (\href{http://cseweb.ucsd.edu/~elkan/clresults.html}{Results of the KDD'99 Classifier Learning Contest}). 

As a comparison, the  BCOPS achieves almost the same predictive power as the vanilla RF for the observed classes while refrains from making predictions for most of the samples from the novel attack types:  the coverage for the normal and intrusion samples using  BCOPS are 99.940\% and 99.942\%, and we assign correct labels to 99.941\% of all samples from the observed classes while refrain from making predictions for 97.3\% of the unobserved intrusion types.

\begin{table}
\centering
\caption{\em Network intrusion results.  The column names are the prediction sets from the BCOPS, and the row names are the true class labels. In each cell, the upper half is the number of samples falling into the category and the lower half is the prediction accuracy from RF for sample in this category. For example, the cell in the column ``normal+intrusion" and row "normal" describes the number of normal connections with BCOPS prediction set contains both normal and at least one intrusion label (upper half) and the prediction accuracy based on RF for these samples (lower half).}
\label{tab:intrusion1} 
\begin{tabular}{|l|l|l|l|l|}
  \hline
\backslashbox{Class}{\thead{BCOPS \\Label}}& normal & intrusion &\thead{normal\\+\\intrusion}& abstention \\ 
  \hline
normal & \backslashbox{1.000}{48635} & \backslashbox{NA}{0 \;\;\;\;\;\;\;\;\;\;\;} & \backslashbox{NA}{0\;\;\;\;\;\;\;\;\;} & \backslashbox{0.931}{29} \\  \hline
  observed intrusions &  \backslashbox{NA\;\;\;\;}{0}& \backslashbox{1.000}{193858} &\backslashbox{NA\;\;\;\;\;}{0\;\;\;\;\;}  & \backslashbox{0.991}{113} \\ \hline
  other intrusions &\backslashbox{ 0.164}{171} & \backslashbox{1.000}{1\;\;\;\;\;\;\;\;\;\;} & \backslashbox{ NA}{0 \;\;\;\;\;\;\;}  & \backslashbox{ 0.471}{8581}  \\ \hline
\end{tabular}
\end{table}

\section{Discussion}
\label{sec:discussion}
In this paper, we propose a conformal inference method with a balanced objective for outlier detection and class label prediction. The conformal inference provides finite sample, distribution-free validity, and with a balanced objective, the proposed method has achieved good performance at both outlier detection and the class label prediction.  Moreover, we propose a method for evaluating the outlier detection rate. Simple as it is, it has achieved good performance in both the simulations and the real data examples in this paper. Here, we also discuss some potential future work:
\begin{enumerate}
\item One extension is to consider the case where  just a single new observation is available. In this case, although we have little information about $f_{test}(x)$, we can still try to design objective functions that lead to good predictions for samples close to the training data set while accepting our ignorance at locations where the training samples are rare. For example, we can use a truncated version of $f(x)$ and let $\mu(x) = f(x)\mathbbm{1}_{f(x)\geq c}+c\mathbbm{1}_{f(x)<c}$ for some properly chosen value $c$. 
\item With just a single new observation available, another interesting question is to ask whether we can use this one sample to learn the direction separating the training data and the new sample. To prevent over-training, especially in high dimension, we can incorporate this direction searching step into the conformal prediction framework. If the direction learning step is too complicated, it will introduce too much variance, if there is no structure learning step, the decision rule can suffer in high dimension even if the problem itself may have simple underlying structure.
\item Another useful extension is to make BCOPS robust to small changes in the conditional distribution $f(x|y)$. The problem is not identifiable without proper constraints.  If we allows only for the transformation $x_{test}\leftarrow ax+b$ from the training data distribution to the test data distribution for a reasonable simple $a$ and $b$ \citep{zhang2013domain}, we may develop a modified BCOPS that is robust to small perturbation to features. 
\end{enumerate}
\medskip {\bf Acknowledgement}. The authors would like to thank Larry Wasserman, who suggested  to us that conformal
prediction might  be a good approach for outlier detection.  Robert Tibshirani was supported by
NIH grant 5R01 EB001988-16 and NSF grant 19 DMS1208164.

\newpage
\appendix
\section{Proofs of Theorems \ref{thm:thm1}-\ref{thm:shift}}
\label{app:BCOPSproof}
Before we prove Theorems \ref{thm:thm1}-\ref{thm:shift}, we first state some Lemmas that will be useful in the proofs.  Let $\hat F_k$ be the empirical distribution of $F_k$ for given samples based on the context. In this paper, when it comes to the empirical estimates, they are are estimated with samples of size $cn$ for a constant $c > 0$ that depends on the context.

\begin{lemma}
\label{lem:lemma0}
Let $G$ and $\hat G$ be the CDF and empirical CDF of a univariate variable in $\real$.  Let $V_{n} = \sup_t |G(t) - \hat G(t)|$, then, for a large enough constant $B$, we have $P(V_{n} \leq B\sqrt{\frac{\log n}{n}})\rightarrow 1$.
\end{lemma}
\begin{proof}
Lemma \ref{lem:lemma0} is a result of  the classical empirical process theory (\cite{wellner2013weak}), see also Lemma C.1 in  \cite{lei2013distribution}.
\end{proof}
\begin{lemma}
\label{lem:lemma1}
Let $g(t)$ be the density of the univariate variable $t\in \real$. Let $\hat g(t)$ be its Gaussian kernel density estimation with bandwidth $h_n>\frac{\log n}{n}$.  Suppose $g(t)$ is bounded and H$\ddot{o}$lder continuous with exponent $1\geq \alpha > 0$, then, there exists a large enough constant $B$, such that with probability at least $1-\frac{1}{n}$, we have
\[
\|g(t) - \hat g(t)\|_{\infty} < B(h^{\frac{\alpha}{2}}_n+\sqrt{\frac{\log n}{n h_n}})
\]
\end{lemma}
\begin{proof}
Obviously, the Gaussian kernel  $K(z)$ for $z\in \real$ satisfies Assumption 2-3 in \cite{jiang2017uniform} (the spherically symmetric and non-increasing Assumption, and the exponential decay Assumption), and $g(t)$ is bounded (Assumption 1 in  \cite{jiang2017uniform}). Then, Lemma \ref{lem:lemma1} a special case of Theorem 2 in \cite{jiang2017uniform}.
\end{proof}

\subsection{Proof of Theorem \ref{thm:thm1}}
The proof follows the same procedure as in \cite{lei2013distribution}. Let  $A_k$ be the accepted region for class $k$ under the oracle  BCOPS.  Let $R_{n,0} = P_{out}(x\notin A_n)$, $R_{n,1} =\sup_{x\in A_n} |\hat v_{k}(x) - v_k(x)|$ and $R_{n,2} =|Q(\alpha; v_k, F_k) -Q(\alpha; v_k, \hat F_k)|$. Then we have
\begin{align}
\label{eq:proof1}
(\widehat A_k\setminus A_{k})\cap A_n = &\{x| x\in A_n, v_k(x) < Q(\alpha; v_k, F_k), \hat v_k(x) \geq Q(\alpha; \hat v_k, \hat F_k) \}\nonumber\\
\subseteq& \{x|Q(\alpha; v_k, F_k) - 2R_{n,1}-R_{n,2}\leq v_k(x) <Q(\alpha; v_k, F_k)\}
\end{align}
and
\begin{align*}
\int_{\widehat A_k\setminus A_{k}}f_{test}(x)dx &\leq \int_{(\widehat A_k\setminus A_{k})\cap A_n}f_{test}(x)dx+R_{n,0}\\
&\leq \int_{(\widehat A_k\setminus A_{k})\cap A_n }\frac{v_k(x)}{[Q(\alpha; v_k, F_k) - 2R_{n,1}-R_{n,2}]_+}f_{test(x)}d(x)+R_{n,0} \\
&\leq \frac{P_{F_k}((\widehat A_k\setminus A_{k})\cap A_n )}{[Q(\alpha; v_k, F_k)  - 2R_{n,1}-R_{n,2}]_+}+R_{n,0}
\end{align*}
By Assumption \ref{ass:ass2}, for a large enough constant $B_1$, we have 
\begin{equation}
\label{eq:proof2}
P(R_{n,1} \leq B_1(\frac{\log n}{n})^{\frac{\beta_1}{2}})\rightarrow 0,\; P(R_{n,0} \leq B_1(\frac{\log n}{n})^{\frac{\beta_2}{2}})\rightarrow 0
\end{equation}
Let $G$ and $\hat G$ be the CDF and empirical CDF of $v_k(x)$.  By Lemma \ref{lem:lemma0}, on the one hand, with probability approaching 1, for any constant $\delta$ and a constant $B_2$ large enough, we have
\[
|\hat G(Q(\alpha-\delta; v_k, F_k))-(\alpha-\delta)| \leq B_2\sqrt{\frac{\log n}{n}}
\]
On the other hand, by Assumption \ref{ass:ass1}, we have
\[
\delta \geq c_1 |Q(\alpha; v_k, F_k) -Q(\alpha\pm\delta; v_k, F_k) |^\gamma
\]
In other words, with probability approaching 1, the following is true
\[
Q(\alpha-B_2\sqrt{\frac{\log n}{n}}; v_k,  F_k) \leq Q(\alpha; v_k, \hat F_k) \leq Q(\alpha+B_2\sqrt{\frac{\log n}{n}}; v_k, F_k)
\]
and
\begin{equation}
\label{eq:proof3}
|Q(\alpha\pm B_2\sqrt{\frac{\log n}{n}}; v_k, F_k) - Q(\alpha; v_k, F_k)| \leq (\frac{B_2}{c_1}\sqrt{\frac{\log n}{n}})^{\frac{1}{\gamma}} 
\end{equation}
Hence,  $R_{n,2}\leq (\frac{B_2}{c_1}\sqrt{\frac{\log n}{n}})^{\frac{1}{\gamma}}$.

For the numerator of equation (\ref{eq:proof1}), we use equations (\ref{eq:proof2})-(\ref{eq:proof3}) and apply Assumption \ref{ass:ass1} again, for a large enough constant $B_3$, we have,
\begin{align*}
&P_{k}((\widehat A_k\setminus A_{k})\cap A_n) \\
\leq& P_{k}(v_{k,\alpha} - 2R_{n,1}-R_{n,2}\leq v(x) \leq v_{k,\alpha})\\
 \leq &(c_2(2R_{n,1}+R_{n,2}))^\gamma \leq B_3(\sqrt{\frac{\log n}{n}}+(\frac{\log n}{n})^{\frac{\beta_1\gamma}{2}})
\end{align*}
For the denominator of equation (\ref{eq:proof1}), when $\alpha$ is a positive constant, since $v_k(x) = \frac{f_k(x)}{f_{test}(x)+f_k(x)}$ is non-zero as long as $f_k(x)$ is non-zero, we must have that $Q(\alpha; v_k, F_k)$ is also a positive constant, and we can always take $n$ large enough, such that
\begin{equation}
\label{eq:denominator}
Q(\alpha; v_k, F_k)  - 2R_{n,1}-R_{n,2} \geq \frac{1}{2}Q(\alpha; v_k, F_k) > 0
\end{equation}
and that for a constant $B$ large enough, with probability approaching 1:
\begin{equation}
\label{eq:numerator}
\int_{\widehat A_k\setminus A_{k}}f_{test}(x)dx \leq \frac{B}{K} ((\frac{1}{n})^{\frac{\min(\gamma \beta_1,\beta_2,1)}{2}})
\end{equation}
Combining eq.(\ref{eq:denominator})-(\ref{eq:numerator}), with probability approaching 1, we have
\[
\int (|\widehat C(x)| - |C(x)|)f_{test}(x) d x \leq \sum_k \int_{\widehat A_k\setminus A_{k}}f_{test}(x)dx\leq B ((\frac{1}{n})^{\frac{\min(\gamma \beta_1,\beta_2,1)}{2}})
\]
with probability approaching 1 for a large enough constant $B$.

\subsection{Proof of Theorem \ref{thm:shift}}
Since $\bSigma$ is invertible with smallest eigenvalue $\sigma_{\min} \geq c >0$ for a constant $c$ and $(\bP^T\bP)^{-1}\bSigma^{-1}\bP^T\rP = \tilde \pi$ by Assumption \ref{ass:ass5}. To show that $\hat \pi = (\hat \bP^T\hat \bP)^{-1}\hat \bP^T\hat \rP$, where $\hat \bP$ and $\hat \rP$ are the empirical versions of $\bP$ and $\rP$,  it is sufficient to show
\[
\hat \bP_{l,k}\overset{p}{\rightarrow} \bP_{l,k},\;\; \forall k = 1,\ldots, K, \bold{R},\;l = 1,\ldots,K
\]
In section \ref{sec:theory_shift}, we have only defined $\bP_{l,k}$ for $k =1,\ldots, K$, here we include the class $\bold{R}$ as well following the same definition: $\bP_{l, \bold{R}} = P_{\bold{R}}(\eta_l(x)\in S_l)$. Recall the definition of $\hat \bP_{l,k}$:
\[
\hat \bP_{l,k} \coloneqq P_{\hat F_k}(\eta_l(x)\in \hat S_l) = ( \hat \bP_{l,k}-\tilde \bP_{l,k})+\tilde \bP_{l,k}
\]
where 
\begin{itemize}
\item $\tilde \bP_{l,k} \coloneqq P_{F_k}(\eta_l(x)\in \hat S_l)$.
\item $\hat S_l \coloneqq \{t: \hat g_{l,l}(t) \geq \hat g_{l,\zeta}\}$, $\hat g_{l,l}(t)$ is the kernel estimation of $g_{l,l}(t)$ and $\hat g_{l,\zeta}\coloneqq Q(\zeta; \hat g_{l,l}\circ \eta_l, \hat F_l)$.
\end{itemize}
We first show that 
\[
\tilde \bP_{l,k}\overset{p}{\rightarrow} \bP_{l,k},\;\; \forall k = 1,\ldots, K, \bold{R},\;l = 1,\ldots,K
\]
Let $\Delta = \bP_{l,k}-\tilde \bP_{l,k} = \Delta_1-\Delta_2$ where $\Delta_1 = \int_t g_{l,k}(t)\mathbbm{1}_{\hat g_{l,\zeta}\leq \hat g_{l,l}(t) \leq g_{l,\zeta}}$ and $\Delta_2 = \int_t g_{l,k}(t)\mathbbm{1}_{g_{l,\zeta}\leq \hat g_{l,l}(t) \leq \hat g_{l,\zeta}}$. We now prove $\Delta_1\overset{p}{\rightarrow} 0$. Let $R_{n,1} = \|\hat g_{l,l}(t) - g_{l,l}(t)\|_\infty$ and $R_{n,2} = |\hat g_{l,\zeta} - g_{l,\zeta}|$.  Then, for a large enough constant $B$, we have
\begin{align}
\label{eq:shift0}
\Delta_1 &=\int_t g_{l,k}(t)\mathbbm{1}_{\hat g_{l,\zeta}+\hat g_{l,l}(t) - g_{l,l}(t)\leq g_{l,l}(t) \leq g_{l,\zeta}+\hat g_{l,l}(t) - g_{l,l}(t)}\nonumber\\
&\leq \int_t g_{l,k}(t)\mathbbm{1}_{g_{l,\zeta}+\hat g_{l,l}(t) - g_{l,l}(t)-R_{n,2}\leq g_{l,l}(t) \leq g_{l,\zeta}+\hat g_{l,l}(t) - g_{l,l}(t)}\nonumber\\
& \leq \left(\frac{\max_t g_{l,k}(t)}{[g_{l,\zeta}-R_{n,1}-R_{n,2}]_+}\right) P_l(g_{l,\zeta}+\hat g_{l,l}(t) - g_{l,l}(t)-R_{n,2} \leq g_{l,l}(t) \leq g_{l,\zeta}+\hat g_{l,l}(t) - g_{l,l}(t))\nonumber\\
& \leq  B\left(\frac{\max_t g_{l,k}(t)}{[g_{l,\zeta}-R_{n,1}-R_{n,2}]_+}\right) R^{\gamma}_{n,2}
\end{align}
The last step is a result of Assumption \ref{ass:ass4}. Observe that
\begin{itemize}
\item  Under Assumption \ref{ass:ass3}, let the constant $\alpha> 0$  be the H$\ddot{o}$lder exponent for $g_{l,l}(t)$, we apply Lemma \ref{lem:lemma1} and have that for a large enough constant $B$:
\begin{equation}
\label{eq:shift1}
P(R_{n,1}\geq B(\sqrt{\frac{\log n}{n h_n}}+h_n^{\frac{\alpha}{2}}))\rightarrow 1\Rightarrow R_{n,1}\overset{p}{\rightarrow} 0
\end{equation}
\item Notice that we also have that $\forall \delta \in (-\delta_0, \delta_0)$:
\begin{equation}
\label{eq:shift2}
|Q(\zeta-\delta;  g_{l,l}, \hat F_l)) -\hat g_{l,\zeta - \delta} | \leq R_{n,1}\overset{p}{\rightarrow} 0
\end{equation}
Thus, we have 
\[
R_{n,2} = |\hat g_{l,\zeta}-Q(\zeta;  g_{l,l}, \hat F_l)+Q(\zeta;  g_{l,l}, \hat F_l)-g_{l,\zeta}| \leq R_{n,1}+|Q(\zeta;  g_{l,l}, \hat F_l))-g_{l,\zeta}|
\]
Let $G$ and $\hat G$ be the CDF and empirical CDF of  $g_{l,l}(\eta_l(x))$ in class $l$. Apply Lemma \ref{lem:lemma0},  we know that  there exists a constant $B$, such that $\forall \delta \in (-\delta_0, \delta_0)$:
\begin{align*}
&|G(Q(\zeta-\delta;  g_{l,l}, \hat F_l)) -(\zeta - \delta)| \leq  \|G-\hat G\|_\infty \leq B\sqrt{\frac{\log n}{n}} \\
\end{align*}
Under Assumption \ref{ass:ass4}, we have
\[
\delta \geq c_1 |Q(\zeta-\delta;  g_{l,l},F_l) - g_{l,\zeta} |^\gamma
\]
In other words, with probability approaching 1, for a large enough constant $B$, the following is true
\[
Q(\zeta-B\sqrt{\frac{\log n}{n}};g_{l,l}, F_l) \leq Q(\zeta;g_{l,l},  \hat F_l)  \leq Q(\zeta+B\sqrt{\frac{\log n}{n}};g_{l,l}, F_l)  \Rightarrow |g_{l,\zeta} - Q(\zeta; g_{l,l}, \hat F_l)| \leq B(\frac{\log n}{n})^{\frac{1}{2\gamma}}
\]
Hence,  $R_{n,2}\leq R_{n,1}+B(\frac{\log n}{n})^{\frac{1}{2\gamma}}$ for a large enough constant $B$.
\end{itemize}
Hence, as $n\rightarrow \infty$, we have $g_{l,\zeta}-R_{n,1}-R_{n,2}\rightarrow g_{l,\zeta}$ is a positive constant. Combine the above analysis with eq.(\ref{eq:shift0}) and that the density $g_{l,k}$ is bounded, for a large enough constant $B$, we have
\[
\Delta_1 \leq B\left(R_{n,1}+(\frac{\log n}{n})^{\frac{1}{2\gamma}}\right)^{\gamma}\rightarrow 0
\]
By symmetry, $\Delta_2\overset{p}{\rightarrow} 0$ will follow the same argument argument. Hence, we have $\Delta\overset{p}{\rightarrow} 0$. Next, we show that
\[
| \hat \bP_{l,k}-\tilde \bP_{l,k} |\overset{p}{\rightarrow} 0
\]
Let $G$ and $\hat G$ be the CDF and empirical CDF of  $ g_{l,l}(\eta_l(x))$ in class $k$, this is true with the argument below:
\begin{align*}
| \hat \bP_{l,k}-\tilde \bP_{l,k} |& =  |P_{\hat F_k}(\hat g_{l,l}(\eta_l(x))\leq \hat g_{l,\zeta}) - P_{ F_k}(\hat g_{l,l}(\eta_l(x)) \leq \hat g_{l,\zeta})|\\
& \leq \max\left(|\hat G(\hat g_{l,\zeta}+R_{n,1}) - G(\hat g_{l,\zeta}-R_{n,1})|, |\hat G(\hat g_{l,\zeta}-R_{n,1}) - G(\hat g_{l,\zeta}+R_{n,1})|\right)\\
&\leq \|\hat G - G\|_\infty+|G(\hat g_{l,\zeta}+R_{n,1}) -G(\hat g_{l,\zeta}-R_{n,1}) |
\end{align*}
By Lemma \ref{lem:lemma0}, we have that $ \|\hat G - G\|_\infty < B \sqrt{\frac{\log n}{n}}$ for a large enough constant $B$ with probability approaching 1. Following the same argument as eq.(\ref{eq:shift0}), we know that for a large enough constant $B$:
\[
|G(\hat g_{l,\zeta}+R_{n,1}) -G(\hat g_{l,\zeta}-R_{n,1}) | \leq B\frac{\max g_{l,k}(t)}{[g_{l,\zeta}-R_{n,2}-R_{n,1}]_+}(2R_{n,1})^{\gamma}\overset{p}{\rightarrow} 0
\]
Hence, we also have $| \hat \bP_{l,k}-\tilde \bP_{l,k} |\overset{p}{\rightarrow} 0$ and we have thus proved our statement.
\bibliographystyle{agsm}
\bibliography{outliers}
\end{document}